\documentclass[onecolumn]{IEEEtran}

\usepackage{macro}

\begin{document}

\title{The Capacity of Causal Adversarial Channels}

\author{\IEEEauthorblockN{%
Yihan Zhang, 
Sidharth Jaggi, 
Michael Langberg, 
Anand D.~Sarwate
}%
\thanks{Y.~Zhang 
is with the 
Institute of Science and Technology Austria, 
Vienna, Austria, 
\texttt{zephyr.z798@gmail.com}. 
S.~Jaggi 
is with the 
University of Bristol, 
Bristol, UK, 
\texttt{sid.jaggi@bristol.ac.uk}. 
Michael Langberg
is with the 
University at Buffalo, Buffalo, NY, USA, 
\texttt{mikel@buffalo.edu}. 
Anand D.~Sarwate 
is with  
Rutgers, The State University of New Jersey, 
Piscataway, NJ, USA, 
\texttt{anand.sarwate@rutgers.edu}.
}%
\thanks{The work of ADS and ML was supported in part by the US National Science Foundation under awards CCF-1909468 and CCF-1909451.}
}

\maketitle

\begin{abstract}
We characterize the capacity for the discrete-time arbitrarily varying channel with discrete inputs, outputs, and states when
(a) the encoder and decoder do not share common randomness, 
(b) the input and state are subject to cost constraints, 
(c) the transition matrix of the channel is deterministic given the state, 
and 
(d) at each time step the adversary can only observe the current and past channel inputs when choosing the state at that time.
The achievable strategy involves stochastic encoding together with list decoding and a disambiguation step. The converse uses a two-phase ``babble-and-push'' strategy where the adversary chooses the state randomly in the first phase, list decodes the output, and then chooses state inputs to symmetrize the channel in the second phase. These results generalize prior work on specific channels models (additive, erasure) to general discrete alphabets and models.
%\mikel{ML: Do we want to say something regarding the 1-dimensional nature of the state space?}
\end{abstract}

\maketitle

\tableofcontents

\newpage

\section{Introduction}
\label{sec:intro}

% !TEX root = causal_general.tex

In introductory courses on information theory and coding theory students encounter two basic models for communication channels. The Shannon-theoretic model~\cite{shannon_mathematical_1948} for memoryless channels treats the effect of the channel as random, where each input symbol is transformed to an output symbol through the same conditional distribution at each time step. Two canonical examples are the binary symmetric channel (BSC) and binary erasure channel (BEC). With high probability, for sufficiently large $n$, a BSC flips close to $pn$ bits for a codeword of blocklength $n$ and the probability of error is \emph{average-case}, measured over the randomness in the channel. By contrast, in the basic coding theory model, errors and erasures are modeled as \emph{worst-case}: for a blocklength $n$ the goal is to design a code which can correct any pattern of $pn$ errors or erasures. 

One way to understand the difference between these models is to frame them both in the context of  arbitrarily varying channels (AVCs)~\cite{blackwell-avc-1960} under constraints~\cite{csiszar-narayan-it1988,csiszar-narayan-it1988-2}. In the AVC there are three participants: Alice (the transmitter/encoder), Bob (the receiver/decoder), and James (an adversarial jammer).
When communicating over an AVC, Alice encodes her message into a codeword $\vx$ of blocklength $n$ and James can choose an equal-length vector of channel states $\vs$. The output $\vy$ is formed by applying a channel law $W_{\bfy|\bfx,\bfs}(y | x, s)$ letter-by-letter to $(\vx,\vs)$. The difference between the two classical communication models can be captured by modeling the information James has about the transmitted codeword. The Shannon-theoretic model is similar to an \emph{oblivious adversary} who must choose $\vs$ without any knowledge of $\vx$. The coding-theoretic model is similar to a \emph{omniscient adversary} in which James can choose $\vs$ as a function of the entire codeword $\vx$.

Once we frame the difference between average and worst case models in terms of the AVC, a variety of ``intermediate case'' models become natural by changing what James can know about the transmitted codeword. In this paper we consider one such model: the \emph{causal (or online) adversary} in which James chooses the channel state $\vs(t)$ at time $t$ based on knowledge of the current and past inputs $(\vx(1), \vx(2), \ldots, \vx(t))$. The online adversary is a special case of the \emph{delayed adversary}~\cite{langberg_binary_2009,dey_coding_2010}. 
%In these models James observes a delayed version of the transmitted codeword, i.e., at time $t$ he observes inputs up to time $t - \Delta$: our model corresponds to delay $\Delta = 0$.   For longer delay $\Delta = \delta n$ the capacity is the same as the oblivious adversary for additive channels~\cite{DeyJLS:10isit}.

%\subsection*{Related work}

%Due to space constraints, we cannot provide a full treatment of the AVC literature but instead focus on those works most relevant to the current paper.  
%The online adversary is a special case of the \emph{delayed adversary}~\cite{langberg_binary_2009,dey_coding_2010}. In these models James can observe a delayed version of the transmitted codeword, so at time $t$ he can observe inputs up to time $t - \Delta$: our model corresponds to delay $\Delta = 0$.  While for longer delay $\Delta = \delta n$ the capacity is the same as the oblivious adversary~\cite{DeyJLS:10isit}, there is a difference between $\Delta = 0$ and $\Delta = 1$. 
Much of the prior work on causal adversaries deals with specific channel models. Capacity results for special cases of causal adversaries with ``large alphabets''~\cite{dey-delayed-it2013},  the erasure setting~\cite{bassily-smith-soda2014,chen-et-al,chen2019capacity}, the bit-flip/symbol-error setting~\cite{dey-causal-it2013,chen2019capacity}, and the quadratically-constrained scenario~\cite{li2018quadratically} are known. Other related channel models explored include settings with a memoryless jammer~\cite{mazumdar2014capacity}, and bit-flip and erasure models in which the channel is not state-deterministic but James can observe the channel output~\cite{VinayakRL:21snooping}.

In this work we focus on AVCs with finite input, state, and output alphabets which are \emph{state-deterministic}, meaning the the channel output $\vy(t)$ at each time $t$ is a deterministic function of $\vx(t)$ and $\vs(t)$. In such models, James can compute the channel output.  
%The goal of the present work
Our goal is to establish capacity results for general state-deterministic AVC models with cost constraints. %As we will see, this requires a more complex analysis for both the achievability and converse. %\ads{expand on what is more complicated here in an understandable manner}

%\ads{Paragraph on the main proof ideas/insights? Add after the main emphasis points in the proof sketch are settled.}

After defining our model in Section~\ref{sec:model}, and key concepts in Section~\ref{sec:achievability}, we present an overview of our capacity analysis in Section~\ref{sec:proofsketch}. 
At a high level, both our achievability and converse proofs follow those appearing in \cite{dey-causal-it2013,chen2019capacity} addressing the causal bit-flip/symbol-error setting.
The main technical contribution in this work thus lies in the highly non-trivial nature of expanding the concepts and analysis in \cite{dey-causal-it2013,chen2019capacity} to fit the generalized model of AVCs with both state and input constraints.
We highlight the major challenges of analyzing general AVCs and the tools used to overcome these challenges in the overview of Section~\ref{sec:proofsketch}.

%For the online model we consider here, Haviv and Langberg}~\cite{haviv-langberg-isit2011} show that the classical Gilbert-Varshamov bound~\cite{gilbert_comparison_1952,varshamov_estimate_1957} for the omniscient case can be improved. Improved upper bounds were shown by Dey et al~\cite{dey-causal-it2013}. 

\section{Proof skecth}
\label{sec:proofsketch}
% !TEX root = causal_general.tex

The converse argument, broadly speaking, generalizes the two-phase ``babble-and-push" jamming strategy and analysis introduced in~\cite{dey-causal-it2013,chen2019capacity} in the bit-flip/symbol-error setting, with some novelties required for the general setting. In particular the recently developed Generalized Plotkin bound~\cite{wbbj-2019-omni-avc} plays a critical role, since via combinatorial arguments it guarantees that regardless of Alice's encoding strategy, with positive probability the joint type of randomly sampled pairs of codeword suffixes will fall within a restricted range (just as the classical Plotkin bound for binary codes can be used to argue that for any code of positive rate, for any $\varepsilon > 0$ the Hamming distance of a randomly chosen pair of codewords does not exceed $\frac{n}{2}(1+\varepsilon)$ with  probability $\Omega(\varepsilon)$). In particular, the Generalized Plotkin bound ensures that with positive probability the joint type of a randomly sampled pairs of codewords (or, in the converse argument here, randomly sampled codeword suffixes) is a convex combination of product distributions. The implication of these types of bounds is that if James just wishes his attack to work with positive probability bounded away from zero, rather than probability 1, then in general he can use a significantly less costly state sequence (just as in the binary setting, with  probability $\Omega(\varepsilon)$ he only needs to flip about $\frac{n}{4}(1+\varepsilon)$ bits to confuse Bob between the truly transmitted codeword and some other codeword which differ in $\frac{n}{2}(1+\varepsilon)$ bits, rather than $\frac{n}{2}$ as in the worst-case when the two codewords differ in every bit). 

Let $\varepsilon$ be an arbitrarily small positive constant that will help specify various slack parameters in what follows. In particular, James' attack strategy operates on consecutive length-$\varepsilon n$ chunks of Alice's codewords, and hence operates on $K = 1/\varepsilon$ chunks. In James' attack below, summarized first for the special case when Alice employs a deterministic coding strategy (i.e., her transmission is a deterministic function of the message she wishes to transmit to Bob). Steps \ref{step:subcode-extraction} and \ref{step:choose-jamming-distr} correspond to James analyzing Alice and Bob's chosen codebook $\cC$ prior to transmission, and steps \ref{step:babble} and \ref{step:push} comprise respectively the babble and push phases during the actual transmission.
\begin{enumerate}
	\item {\it \underline{Sub-code extraction}:} 
	\label{step:subcode-extraction}
	First, given Alice's codebook ${\cC}$ of rate $R$, James finds a subcode $\cC' \subseteq \cC$ which has both the following two properties: 
	\begin{enumerate}
		\item[(a)] 	The subcode $\cC'$ contains a constant fraction of the codewords in $\cC$, and
		\item[(b)] Each codeword in $\cC'$ is {\it chunk-wise $\delta$-approximately constant composition}, i.e., for each codeword $\vx \in \cC'$ and each chunk with index $u \in \{1,\ldots, K\}$, the {\it $u$-th chunk-wise type of $\vx$}, denoted $\type_{\vx^{(u)}}$ and defined as the type of $\vx$ restricted to symbols in the $u$-th chunk, differs in the $\ell_\infty$ norm from a given distribution $P_{\bfx|\bfu = u}$ by at most $\delta$. 
	\end{enumerate}
 Such a sub-code can always be extracted for the following reason. There are $1/\epsilon$ chunks, and so the simplex of all possible chunkwise-types of a length-$n$ code can be $\delta$-approximated by a $\delta$-net with $\cO((1/\delta)^{|\cX|/\epsilon})$ elements. Hence picking the cell of the $\delta$-net with the most codewords in it gives us the subcode $\cC'$ with a constant fraction (at least $\Omega(\delta^{|\cX|/\epsilon})$-fraction) of the codewords in $\cC$. Let this cell correspond to the set $\{P_{\bfx|\bfu = u}\}_{u=1}^K$ of distributions. Going forward, James' attack is only guaranteed to succeed with positive probability for codewords $\vx \in \cC'$. Despite this restriction, due to (b), this nonetheless happens with strictly positive probability, bounded away from zero independently of the blocklength $n$. 
	\item {\it \underline{Choice of attack parameter/jamming distributions}:} 
	\label{step:choose-jamming-distr}
	James then simultaneously selects a {\it prefix-length parameter} $\alpha \leq 1$ (w.l.o.g.\ it is assumed that $\alpha K$ is an integer -- if not, $\alpha$ can be appropriately quantized), and two corresponding sets of jamming distributions:
	 \begin{enumerate}
\item [(a)] {\it \underline{Prefix-babbling distributions}:} These comprise of a set of $\alpha K$ distributions $\{V_{\bfs|\bfx,\bfu = u}\}_{u=1}^{\alpha K}$ that James will use in the corresponding prefix chunks, as described in~\ref{step:babble} below. 
\item [(b)] {\it \underline{Suffix-symmetrization distributions}:} These comprise of a set of $\alpha K$ distributions $\{V_{\bfs|\bfx,\bfx',\bfu = u}\}_{u=\alpha K + 1}^{K}$ that James will use in the corresponding suffix chunks, as described in~\ref{step:push} below. 
	 \end{enumerate}
	 In particular, these two sets of distributions are required to satisfy both the following:
	 \begin{enumerate}
	 	\item [(i)] {\it \underline{State-feasibility}:}\label{step:state-feas} The prefix-babbling and suffix-symmetrization jamming distributions are required to jointly satisfy a certain {\it state-feasibility condition} -- see Definition~\ref{def:feas-jam-dist}. This condition can be thought of as ensuring that the overall jamming vector $\vs$ that James chooses satisfies the state cost constraint imposed by the AVC by separately computing the cost of the prefix jamming sequence $(\vs(1),\ldots,\vs(\alpha n) )$ induced by the prefix-babbling distributions, and the cost of the suffix jamming sequence $(\vs(\alpha n+1),\ldots,\vs(n) )$ induced by the suffix-pushing distributions. It is important to highlight here that due to the Generalized Plotkin bound~\cite{wbbj-2019-omni-avc} the suffix state cost only needs to be computed with respect to $P_{\bfx,\bfx'}$ distributions that are convex combinations of product distributions.
	 	\item [(i)] {\it \underline{Suffix-symmetrizability}:} The suffix-symmetrization jamming distributions are required to satisfy a certain {\it suffix-symmetrizability condition} -- see $\cV$ in Definition~\ref{def:symm-dist}. This condition can be thought of ensuring that the suffix $(\vy(\alpha n+1),\ldots,\vy(n) )$ observed by Bob was equally likely to have been generated by either the codeword suffix $\vx^{>\alpha}  = (\vx(\alpha n+1),\ldots,\vx(n) )$ or the codeword suffix $\vx'^{>\alpha}  = (\vx'(\alpha n+1),\ldots,\vx'(n) )$ (with the jamming distribution $V_{\bfy|\bfx,\bfx'}$ acting on the pair of codeword suffixes $(\vx^{>\alpha},\vx'^{>\alpha})$ in the first case, and on the pair $(\vx'^{>\alpha},\vx^{>\alpha})$ in the second case -- note the reversal/symmetrization in the two scenarios).
	 \end{enumerate}
	 After this pre-transmission analysis, once Alice starts transmitting, James proceeds with his jamming attack as follows.
		\item {\it \underline{Prefix-babbling phase}:}\label{step:babble} For the initial $\alpha K$ chunks (henceforth called the {\it $\alpha$-prefix of $\vx$}, or just the prefix) James uses a ``babble" attack. That is, for each time index $t$ in chunk $u \leq \alpha K$ he selects his jamming symbols $\vs(t)$ via the probabilistic map $\{V_{\bfs|\bfx, \bfu =u}\}_{u=1}^{\alpha K}$ applied symbol-by-symbol to the corresponding observed transmission $\vx(t)$. For these chunks in the prefix James thereby induces effective DMCs $\{W_{\bfy|\bfx,\bfu = u}\}_{u=1}^{\alpha K}$ from Alice to Bob. If the rate $R$ of Alice's code is too high (i.e., is not $\delta$-good -- see Definition~\ref{def:goodness-ach}), then for some $\alpha \leq 1$ there will be a chunk number $\alpha K$  such that in addition to James' jamming sequence $\vs$ satisfying the state-feasability and suffix-symmetrizability conditions described in 2) above, he can also force the ``Large prefix-list" condition described below.
	\begin{itemize}
	\item {\it \underline{Large prefix-list}:} For the prefix, the normalized cumulative mutual information $\sum_{i=1}^{\alpha K} \frac{1}{K} I(P_{\bfx|\bfu = u},W_{\bfy|\bfx,\bfu = u})$ (repeated in Definition~\ref{def:cumulative-mutual-info}) is less than the rate $R$ of Alice's code. Via standard information-theoretic arguments such as strong converses, this can be used to argue that at this point in the transmission Bob must still have a large set (say of size $2^{\Omega(n)}$) of potential messages that are compatible with Alice's coebook $\cC$, his prefix observation $\vy^{\leq \alpha}$ and the prefix channels $\{W_{\bfy|\bfx,\bfu = u}\}_{u=1}^{\alpha K}$ induced by James.  
	\end{itemize}
	 \item \label{step:push} {\it \underline{Suffix-pushing phase}:} Since the channel $W_{\bfy|\bfx,\bfs}$ is state-deterministic (this is the precise reason we need to restrict to state-deterministic AVCs the class of channels for which we can claim our achievability and converse arguments match), therefore at this $\alpha n$ point James can compute the same large prefix-list that Bob would have computed. He then randomly selects some $m'$ in this prefix list, and sets the ``spoofing suffix" $\vx'^{>\alpha}_{m'} $ as the codeword suffix corresponding to this $m'$, and for any time-index $t$ in suffix chunk $u$,  applies the corresponding suffix symmetrizing distribution $V_{\bfs|\bfx,\bfx',\bfu = u}(\cdot| \vx(t),\vx'(t))$ symbol-by-symbol on the $\vx(t)$ he (causally) observes and the $\vx'(t)$ in his chosen spoofing suffix.
\end{enumerate}
Via analysis paralleling but significantly generalizing to a general class of AVCs the symbol-error analysis in~\cite{chen2019capacity} it can be shown that for any code with rate that is not $\delta$-good the strategy outlined and motivated above results, with probability bounded away from zero, in James being able to cause a decoding error by Bob. We emphasize again that incorporating the recently discovered Generalized Plotkin bound was a key step in this converse, without which any converse would have resulted in a significantly weaker rate-bound for general AVCs. Further, by the synthesis of information-theoretic and coding-theoretic techniques introduced in~\cite{dey-causal-it2013}, these arguments can be generalized even to settings where Alice uses a stochastic encoder. 
% Therefore, by choosing among all possible choices of attack parameter/jamming distributions in Step \ref{step:choose-jamming-distr} above the one that attains the infimum of the ``not $\delta$-good" rates, we can ensure that if she wishes to reliably communicate to Bob, \yihan{incomplete sentence}
Therefore, by choosing among all possible choices of feasible attack parameter/jamming distributions in Step \ref{step:choose-jamming-distr} that attains the infimum of the ``not $\delta$-good'' rates, we can ensure that if Alice wishes to reliably communicate to Bob, James will be able to confuse Bob between the truly transmitted message and a spoofing message.

To complement this impossibility argument, we can also demonstrate that for the class of AVCs considered any rate that is $\delta$-good is indeed achievable. Again, this broadly parallels the scheme for symbol-error AVC in~\cite{chen2019capacity}, while significantly generalizing it to a general class of AVCs. In particular, the achievability analysis required novel interpolating arguments to ensure that certain algebraic conditions on distributions ($\delta$-goodness of the code, non-symmetrizability of the suffix) ensures that if Bob follows a certain iterative decoding mechanism, he will, with high probability, be able to decode correctly.

\begin{enumerate}
	\item {\it \underline{Codebook parameter choices}:} First, for each chunk Alice chooses a length-$K$ time-sharing vector $\vu$ and corresponding chunk-wise input distributions $\{P_{\bfx|\bfu = u}\}_{u=1}^{K}$ such that the overall time-averaged input distribution $P_\bfx = \frac{1}{K}\sum_{u=1}^K P_{\bfx|\bfu = u}$ satisfies the input constraint set specified in the AVC problem. Among all possible such choices of $\vu$ and $\{P_{\bfx|\bfu = u}\}_{u=1}^{K}$, Alice chooses the time-sharing vector and chunk-wise input distributions for which the supremum of all $\delta$-good rates is as large as possible, and sets this supremum (minus a $\delta$ slack) as the rate $R$ of her code.
	\item {\it \underline{Chunk-wise stochastic code design}:} Alice then designs a random chunk-wise stochastic code as follows. In particular, for each chunk $u \in \{1,\ldots,K\}$  and each message $m \in \{1,\ldots,2^{nR}\}$, she chooses $2^{\varepsilon^3 n}$ random length-$\varepsilon n$ codeword chunks of type $P_{\bfx|\bfu = u}$. These codeword-chunks for  chunk number $u$ and message $m$ are denoted $\vx^{(u)}_{(m,r^{(u)})}$, where $r^{(u)}$ is an index varying in the set $\{1,\ldots,2^{n\varepsilon^3}\}$. These multiple codeword-chunks for a given message are used, as described below, to ensure that at any time $t$ James always has some uncertainty about the actual codeword to be transmitted in the next chunk even if he knows Alice's message.
	\item {\it \underline{Chunk-wise stochastic encoding}:} Given the message $m$ and the chunk-wise stochastic code described above, for each chunk $u \in \{1,\ldots,K\}$ Alice picks a uniformly random $r^{(u)}$ from the set $\{1,\ldots,2^{n\varepsilon^3}\}$ and transmits the codeword-chunk  $\vx^{(u)}_{(m,r^{(u)})}$. For small $\varepsilon$, note that the entropy of $r^{(u)}$ (which equals $n\varepsilon^3$) is significantly smaller than the chunk-length $n\varepsilon$. This is important in the list-decoding steps described below.
	\item {\it \underline{Chunkwise-iterative decoding}:} On receiving the vector $\vy$, Bob's decoding scheme then proceeds iteratively over $K$ stages in a manner that conceptually parallels the babble-and-push jamming strategy described above. (We emphasize that we do not assume that James follows a babble-and-push jamming strategy -- rather, as a consequence of our analysis it turns out to be the case that for the class of AVCs considered the babble-and-push jamming strategy is an extremizing strategy in the sense that it results in the lowest reliable throughput.)\\
	We describe this chunkwise-iterative decoding scheme in more detail now as follows. First, Bob initializes his prefix-chunk counter $\alpha K$ to $1$. Then, for this value of $\alpha$ he performs the following two-step decoding attempt. 
	\begin{enumerate}
	\item {\it \underline{Prefix list-decoding}:} For the given $\alpha$, Bob considers all the tuples $(\alpha,\{V_{\bfs|\bfx,\bfu = u}\}_{u=1}^{\alpha K },\{V_{\bfs|\bfx,\bfx',\bfu = u}\}_{u=\alpha K + 1}^{K})$ which are feasible in the sense of Definition~\ref{def:feas-jam-dist}, and for which $\alpha$ is the first chunk for which the the normalized cumulative mutual information $\sum_{i=1}^{\alpha K} \frac{1}{K} I(P_{\bfx|\bfu = u},W_{\bfy|\bfx,\bfu = u})$ (in the sense of Definition~\ref{def:cumulative-mutual-info}) is {\it more} than the rate $R$ of the code. Via relatively standard list-decoding arguments (adapted to the chunkwise stochastic encoding scheme here) ensure that with overwhelming probability over code-design (indeed, with probability say $1-2^{\Omega(n^2)}$) Bob can decode down to a ``small" list (of size say $\cO(n)$).
	\item {\it \underline{Suffix disambiguation}:} For this value of $\alpha$ and the corresponding list of messages obtained from the prefix, Bob uses the suffix to disambiguate the correct message from the list, by performing a suffix consistency check. This comprises of checking to see if any of the codeword suffixes for any message in the prefix-list is jointly typical with respect to the observed suffix. If more than one message in the prefix list passes this suffix consistency test then Bob declares a decoding error and aborts; if none of the messages in the prefix list pass this consistency test then Bob increments the value of $\alpha$ and tries again; and if a unique message passes this test then Bob outputs that message as the decoded message and terminates. To show that this decoding strategy suffices to attain the capacity, we need to show two things. Firstly we need to demonstrate that the truly transmitted message does indeed pass Bob's suffix consistency check, and secondly that with high probability over the private randomness in the suffix no other message in Bob's prefix list passes the consistency check. The former argument is relatively standard, relying mostly on concentration arguments to show that joint typicality decoding in the suffix results in the true message being suffix-consistent. The latter argument is more delicate -- we first demonstrate that at the true $\alpha$ value, over the private randomness in suffix chunks that is currently unknown to James, the joint types of {\it most} codeword suffixes for the truly transmitted codeword will be jointly typical (with respect to a convex combination of a product distribution) with {\it every} suffix codewords for {\it all} other messages in Bob's prefix message list (here, importantly, we use the fact that the private-randomness rate is vanishing relative to the suffix length, and hence suffix joint-types are closely concentrated around their expectation). This coupled with suffix non-symmetrizability ensures that the true codeword suffix typicality implies that no other message codeword suffix can be jointly typical. Another non-trivial aspect to this analysis is to demonstrate that even for decoding points preceding the ``true" $\alpha$, this decoding rule works in the sense that no other suffix passes the suffix consistency test -- this is done by carefully choosing two interpolating jamming trajectories and using convexity arguments to argue that non-symmetrizability with respect to the true $\alpha$ also implies non-symmetrizability with respect to preceding $\alpha$s. %\sidj{TO BE CONTINUED -- besides describing this stage, give some intuition as to why it work for the correct trajectory, and fails for incorrect trajectories with $\alpha$ less than the true value. Describe the interpolating distributions, and also highlight the reason for the single cost constraint, and also the monotonicity of state cost with mutual information}
	\end{enumerate}
\end{enumerate}

It should be emphasized that, perhaps surprisingly, input distributions $P_{\bfx|\bfu = u}$ that are uniform in the time-sharing variable $\bfu$ do not necessarily attain capacity for general AVCs. For the special case of the finite alphabet symbol-error/erasure channels considered in~\cite{chen2019capacity} the uniform input distribution is indeed optimizing as shown there, but for the quadratically constrained casual model~\cite{li2018quadratically} it can be shown that uniform power allocation is {\it not} capacity-achieving. The reason for this phenomenon, which might seem counter-intuitive given the convexity of mutual information function and the linearity of the symmetrizability condition, is that since Alice has to choose a rate $R$ and design her corresponding codebook distribution to simultaneously be $\delta$-good for any possible choice of $\alpha$ (and the corresponding feasible input distributions), therefore the problem becomes highly non-convex. We are currently actively investigating alternating optimization approaches to computationally tractable well-approximate the capacity function for general AVCs, and simultaneously obtain the corresponding extremizing input/jamming distributions.

\section{Model}
\label{sec:model}
% !TEX root = causal_general.tex

\noindent\textbf{Notation:} alphabets are in calligraphic script (e.g.~$\cX$). In this paper, all alphabets are finite. A boldface letter (e.g.~$\bfx$) is a random variable, with the non-boldface (e.g.~$x$) as its realization. The set $[M] = \{1,2,\ldots, M\}$. Tuples are written with an underline sign and individual entries with the time index in parentheses (e.g.~$\vx = (\vx(1), \vx(2), \ldots, \vx(n))$) and $\vx(1:i) = (\vx(1), \vx(2), \ldots, \vx(i))$). The \emph{type} of a tuple $\type_{\vx}$ is the empirical distribution of $\vx$. The set of all probability distributions on an alphabet $\cX$ is $\Delta(\cX)$. The set of all conditional distributions (randomized maps) from $\cX$ to $\cS$ is $\Delta(\cS | \cX)$. The length $n$ \emph{type class} corresponding to $P \in \Delta(\cX)$ is denoted by $\cT_n(P)= \left\{ \vx \in \cX^n : \type_{\vx} = P \right\}$. For a joint distribution $P_{\bfx, \bfs}$ we write $[P_{\bfx,\bfs}]_{\bfx}$ and $[P_{\bfx,\bfs}]_{\bfs}$ for the marginal distributions of $\bfx$ and $\bfs$.

\subsection{Channels and codes}
\label{sec:channels-codes}

We consider a class of arbitrarily varying channels (AVCs) with cost constraints on the input and state. Our formulation of the cost constraint generalizes the standard definition~\cite{csiszar-narayan-it1988} by modeling the constraint as requiring that the type of the channel input or state belong to a specified set.
%Let $\cX$, $\cS$, and $\cY$ denote the input, state, and output alphabets. 

\begin{definition}[AVC]
\label{def:avc}
An \emph{arbitrarily varying channel} (AVC) is a sextuple $(\cX,\cS,\cY,\lambda_\bfx,\lambda_\bfs,W_{\bfy|\bfx,\bfs})$. 
Here, $\cX,\cS,\cY$ are the input, state and output alphabets, respectively. 
The input and state constraints are specified by $\lambda_\bfx\subset\Delta(\cX) $ and $\lambda_\bfs\subset\Delta(\cS)$, respectively. 
The channel law is 
%given by 
$W_{\bfy|\bfx,\bfs}\in\Delta(\cY|\cX\times\cS)$. 
\end{definition}

Our goal is to communicate one of $M$ messages reliably over this AVC. For a positive integer $M$, let $\cM \coloneqq[M]$ denote all possible messages that the transmitter may send. 

%Let $ M\in\bZ_{\ge1} $ and $ \cM\coloneqq[M] $. The set $\cM$ denotes all possible messages that the transmitter may send. 

\begin{definition}[Codes]
\label{def:codes}
A \emph{code} for a causal AVC $ (\cX,\cS,\cY,\lambda_\bfx,\lambda_\bfs,W_{\bfy|\bfx,\bfs}) $ is a pair $ (\enc,\dec) $. 
Here $ \enc\in\Delta(\cX^n|\cM) $ is a (potentially stochastic) encoder. 
For $ m\in\cM $, we use $ \enc(m)\in\cX^n $ to denote the (possibly random) encoding of $m$. 
Each such an encoding is called a \emph{codeword}. 
The set $ \enc(\cM) $ of all codewords is called the \emph{codebook}, denoted by $\cC$. 
The length $n$ of each codeword is called the \emph{blocklength}. 
The \emph{rate} of $\cC$ is defined as $ R(\cC)\coloneqq\frac{1}{n}\log M $. 
The code is required to satisfy the input constraint: for every $\vx\in\cC $, $ \type_{\vx}\in\lambda_\bfx$. 
The decoder is given by $ \dec\in\Delta(\cM|\cY^n)$. 
We use $ \dec(\vy)\in\cM $ to denote the (potentially random) message output by the decoder given $\vy$. 
\end{definition}

\begin{definition}[Jamming strategies]
\label{def:jamming}
A \emph{jamming strategy} of blocklength $n$ is a set of maps $\jam = (\jam_1,\cdots,\jam_n)$ where $\jam_t \in \Delta(\cS | \cX^n)$ is the jamming function at time $t$. In a \emph{causal jamming strategy},  
%is a jamming strategy in which 
$\jam_t \in \Delta(\cS | \cX^t)$.
\end{definition}

\begin{definition}[Communication over causal/online AVC]
\label{def:causal-avc}
Communication over a  \emph{causal} (a.k.a.\ \emph{online}) AVC 
%is an AVC 
$(\cX,\cS,\cY,\lambda_\bfx,\lambda_\bfs,W_{\bfy|\bfx,\bfs}) $ 
%with 
has the following requirements. 
Let $ \cC $ be a code with an encoder-decoder pair $ (\enc,\dec) $
and $\jam = (\jam_1,\cdots,\jam_n)$ be a causal jamming strategy.
%Let $ \jam = (\jam_1,\cdots,\jam_n) $ denote a (possibly random) jamming functions. 
%Here for each $ i\in[n] $, $ \jam_i\in\Delta(\cS|\cX^i) $. 
We use $ \jam_i(\vx(1),\cdots,\vx(i)) $ to denote the jamming symbol generated by the adversary at time $i$. 
Before communication happens, $ (\enc,\dec)$ are fixed  and revealed to the transmitter \emph{Alice}, the receiver \emph{Bob} and the adversary \emph{James}. James then fixes a causal jamming strategy $\jam = (\jam_1,\cdots,\jam_n)$ which can depend on $(\enc,\dec)$. 
%\ads{I changed this from what was before, which said all three are revealed simultaneously, but James can tailor the jamming strategy to the code!}

%,\jam $ are fixed and revealed to the transmitter \emph{Alice}, the receiver \emph{Bob} and the adversary \emph{James}. 

The code is required to satisfy the input constraint $ \type_\vx\in\lambda_\bfx $ for every $ \vx\in\cC $. 
Once a particular encoding $ \vx $ of a certain message $m$ is transmitted by Alice, James observes $\vx$ \emph{causally}. 
That is, for any $ i\in[n] $, he observes $ \vx(i) $ after observing $ \vx(1),\cdots,\vx(i-1) $. 
Given his causal observation, he %designs a jamming sequence $ \vs\in\cS^n $ causally. 
%That is, for any $ i\in[n] $, 
computes $ \vs(i) = \jam_i(\vx(1),\cdots,\vx(i)) $ which depends only on $ \vx(1),\cdots,\vx(i) $ (and $(\enc,\dec),\cC$ which are known to everyone). 
The channel then outputs $\vbfy$ according to the following distribution 
\begin{align}
\prob{\vbfy = \vy\condon\vbfx = \vx,\vbfs = \vs} &\coloneqq \prod_{i = 1}^n W_{\bfy|\bfx,\bfs}(\vy(i)|\vx(i),\vs(i)). \notag 
\end{align}
Receiving $ \vbfy $, Bob decodes to $ \dec(\vbfy) $. 
\end{definition}

Communication reliability is measured by:
% \emph{the probability of error}. 
\begin{definition}[Error probability]
\label{def:error-prob}
Consider a causal AVC $ (\cX,\cS,\cY,\lambda_\bfs,\lambda_\bfs,W_{\bfy|\bfx,\bfs}) $. 
Let $ (\enc,\dec) $ and $ \jam $ be a coding scheme and a jamming strategy for this channel, respectively. 
Define the \emph{maximum error probability} $P_{\e,\max}(\enc,\dec,\jam) $ as follows  
\begin{align*}
P_{\e,\max}(\enc,\dec,\jam) 
		\coloneqq 
		& \max_{m \in \cM} 
			\sum_{\substack{\wh{m}\in\cM \\ \wh{m}\ne m}} \sum_{\vy\in\cY^n} \sum_{\vs\in\cS^n} \sum_{\vx\in\cX^n} %\sum_{m\in\cM} 
		\enc(\vx|m) \cdot
	% \\
	% & \hspace{0.05cm}
		\left(\prod_{i = 1}^n \jam_i(\vs(i)|\vx(1:i)) W_{\bfy|\bfx,\bfs}(\vy(i)|\vx(i),\vs(i)) \right) \cdot
		\dec(\wh{m}|\vy).
%	\\
%&P_{\e,\avg}(\enc,\dec,\jam) 
%	\\
%	&\hspace{0.2cm}
%	\coloneqq 
%	\max_{\vs\in\cS^n} \frac{1}{M} \sum_{m\in\cM} 
%		\sum_{\vx\in\cX^n} \sum_{\wh{m}\in\cM} \sum_{\vy\in\cY^n} \enc(\vx|m) 
%	\\
%	&\qquad \hspace{0.05cm}
%	\prod_{i = 1}^n \jam_i(\vs(i)|\vx(1:i)) W_{\bfy|\bfx,\bfs}(\vy(i)|\vx(i),\vs(i)) 
%	\dec(\wh{m}|\vy),
\end{align*}
The \emph{average error probability} is defined analogously by averaging over messages $m$:
\begin{align}
P_{\e,\avg}(\enc,\dec,\jam) 
		\coloneqq 
		& \frac{1}{|\cM|} 
			\sum_{\substack{\wh{m}\in\cM \\ \wh{m}\ne m}} \sum_{\vy\in\cY^n} \sum_{\vs\in\cS^n} \sum_{\vx\in\cX^n} \sum_{m\in\cM} 
		\enc(\vx|m) \cdot \notag \\
		& \left(\prod_{i = 1}^n \jam_i(\vs(i)|\vx(1:i)) W_{\bfy|\bfx,\bfs}(\vy(i)|\vx(i),\vs(i)) \right) \cdot
		\dec(\wh{m}|\vy). \label{eqn:def-pe-avg} 
\end{align}
Here we view $ \enc\in\Delta(\cX^n|\cM) $ and $ \dec\in\Delta(\cM|\cY^n) $ as conditional distributions. 
\end{definition}

%\ads{I'm not sure if it would be clearer to just write it (for the ISIT paper) using $ \prob{\hat{\bfm} \ne m}$ for max error error and $\prob{\hat{\bfm} \ne \bfm}$ for average error?}

\subsection{State-deterministic AVCs}
\label{sec:global-assumptions}

We consider a class of AVCs which are state deterministic and have a single cost constraint on the state. 
%We call such AVCs \emph{single constraint state-deterministic} AVCs, or SCSDs for short. 
More precisely, we require the following five assumptions to hold:
\begin{enumerate}
	\item \label{itm:assumption-alpha} All alphabets $ \cX,\cS,\cY $ are finite.\footnote{The quadratically-constrained infinite alphabet setting was considered in~\cite{li2018quadratically}.}

	\item \label{itm:assumption-ip-constr} The input constraint set $ \lambda_\bfx\subset\Delta(\cX) $ is convex. This is a natural restriction -- a non-convex set  $ \lambda_\bfx$ would imply that the encoder is not allowed to time-share between some potential transmissions.

	\item \label{itm:assumption-st-constr} The set $ \lambda_\bfs\subset\Delta(\cS) $ is specified by a \emph{single} constraint:
	\begin{align}
	\lambda_\bfs &\coloneqq \curbrkt{P_\bfs\in\Delta(\cS): \sum_{s\in\cS}P_\bfs(s) B(s) \le \Lambda}, \notag 
	\end{align}
	for some $ {B}\in\bR^{|\cS|} $ and $ \Lambda\in\bR $. 
%		\mikel{We may need to modify the following cost function to allow chunkwise representation} \sidj{I think the model itself doesn't need to change -- chunkwise structure to the codebook would still result in a single linear function of the form below. In the actual code construction it might be more convenient to write the chunkwise form of the cost-function, but that can be done locally.}
	%Let $\cost(P_\bfs)=\sum_{s\in\cS}P_\bfs(s) B(s)$, thus $\lambda_\bfs = \curbrkt{P_\bfs\in\Delta(\cS): \cost(P_\bfs) \le \Lambda}$.
	
	\item \label{itm:assumption-ch} %\yihan{Needed in converse (maybe also in ach.?).}\sidj{I think just the converse.} 
	The channel law $ W_{\bfy|\bfx,\bfs} $ is deterministic, i.e., for every $ (x,s)\in\cX\times\cS $, there is a unique $y\in\cY$ such that $ W_{\bfy|\bfx,\bfs}(y|x,s) = 1 $. 
	Alternatively, we write the channel law as a (deterministic) function $ W\colon\cX\times\cS\to\cY $.~\footnote{Our achievability result actually does not require this restriction -- it holds even if the AVC is not state-deterministic. However, our current converse arguments providing a capacity upper bound asymptotically matching the rate achievable by our achievability scheme rely on state-determinism, since they rely on the jammer being able to predict the channel output resulting from a specific jamming strategy.}
	
	\item \label{itm:assumption-s0} There exists a zero-cost state $s_0\in \cS$ and a one-to-one mapping $\phi:\cY \rightarrow \cX$ for which for every $x$, $W_{\bfy|\bfx,\bfs}(\phi(x)|x,s_0)=1$. This final assumption is rather natural and corresponds to many channel models. It intuitively implies that there is a \emph{stand-down} state for James that on one hand has zero cost and on the other does not corrupt communication at all, e.g., the ``non-erasure'' state in an erasure channel.
	%, or the ``no-flip" state in a bit-flip channel.
%Such a state $s_0$ has no cost to James and, accordingly, does not corrupt Bob's view of $x$ (as $x$ can be decoded from $\phi(x)$).	
%If such a state $s_0$ does not exist for the given channel $W$, one can replace $W$ by $W'$, replace $\cY$ by $\cY'=\cY \cup \{y_x | x \in \cX\}$, and $\cS$ by $\cS'=\cS \cup \{s_0\}$, where $\{y_x | x \in \cX\} \cap \cY = \phi$, for all $x \in \cX$, $y \in \cY$, $s \in \cS$, $W'(y|x,s)=W(y|x,s)$, and for all $x \in \cX$, $W'(y_x|x,s_0)=1$ (the remaining values of $W'(y|x,s)$ are set to 0).
%As James is stronger when one considers the AVC defined by $\cX'$, $\cY'$, and $W'$ when compared to that defined by $\cX$, $\cY$, and $W$, the achievability proof at hand may assume the former without loss of generality.
%\mikel{Need to rethink the argument again - we may actually need to prove that the change does not reduce rate as we cannot use $W'$ in an isolated manner in this part of the proof and not in others ... which will eventually lead to capacity expressions that depend of $W'$.} 
\end{enumerate}

\section{Capacity charatcterization}
\label{sec:achievability}
% !TEX root = causal_general.tex

%\ads{I think this is actually the overview of main results and techniques.}

In order to state our main result we must define a notion of \emph{symmetrizability}~\cite{ericson_exponential_1985,csiszar-narayan-it1988-2} that is appropriate to the causal AVC model. Symmetrizabillity conditions play an important role in characterizing AVC capacities under deterministic coding. Roughly speaking, a channel is symmetrizable if James can, via selecting the state sequence $\vs$, cause the channel to behave like a symmetric two-user multiple-access channel $W_{\bfy|\bfx,\bfx'}(y | x, x')$. Operationally, this means James can select an alternative (``spoofing'') codeword $\vx'$ and use it to create a state $\vs$ such that Bob cannot tell if Alice sent $\vx$ and James chose $\vx'$ or if Alice sent $\vx'$ and James chose $\vx$. 

There are two aspects of the causal AVC which make it tricky to define a notion of symmetrizability. First is the online nature of the adversarial attack: James has to choose the elements of $\vs$ sequentially as opposed to selecting $\vx'$ and then $\vs$. Second is the cost constraint on $\vs$: if we think of the cost as ``power,'' then James faces a power allocation problem. Taken together, James could spend little power at the beginning and more power at the end of the transmission or vice versa. In the former case, Bob can get a good estimate of the message initially but then the channel becomes much worse. In the latter, Bob has a very bad estimate of the message but the channel is less noisy at the end, allowing him to potentially decode to the true message. Since James does not know the transmitted codeword a priori, he has to choose how to allocate the power ``on the fly'' while satisfying the cost constraint.

\subsection{Symmetrizing distributions}

Let $ K\in\bZ_{\ge1} $, $ \cU \coloneqq [K] $ and $ \cA\coloneqq \{0,1/K,2/K,\cdots,1-1/K\} $. 
For $ \alpha\in\cA $, let $ \cU^{\le\alpha} \coloneqq \{1,2,\cdots,\alpha K\} $ and $ \cU^{>\alpha} \coloneqq \{\alpha K+1,\alpha K+2,\cdots,K\} $. 
Let $ P_\bfu \coloneqq \unif(\cU) $. 
Fix $ P_{\bfx|\bfu}\in\Delta(\cX|\cU) $ such that $ \sqrbrkt{P_\bfu P_{\bfx|\bfu}}_{\bfx}\in\lambda_\bfx $. 
% For $ u\in\cU $, let $ P_{\bfx|\bfu = u} \coloneqq P_{\bfx|\bfu = u} $. 

%\subsection{Achievability bound}
%\label{sec:ach-bound}

In our achievable scheme we will use a code in which the total blocklength $n$ is broken into $K$ subblocks. In each block, we characterize the jammer by a single letter channel. 
%\ads{stopped here for now, trying to find the right words to explain the concept for the non-expert.}

%We will characterize jamming schemes by a tuple $(\alpha, V_{\bfs|\bfx, \bfu^{\le\alpha}}, V_{\bfs|\bfx, \bfx', \bfu^{>\alpha}})$ 

\begin{definition}[Feasible jamming distributions]\label{def:feas-jam-dist}
Let $ P_{\bfx|\bfu}\in\Delta(\cX|\cU) $. 
Define, for $ \alpha\in\cA $, 
\begin{align}
\cF_\alpha(P_{\bfx|\bfu}) &\coloneqq \curbrkt{ \begin{array}{l}
(V_{\bfs|\bfx, \bfu^{\le\alpha}}, V_{\bfs|\bfx, \bfx', \bfu^{>\alpha}})\in \Delta(\cS|\cX\times\cU^{\le\alpha})\times\Delta(\cS|\cX^2\times\cU^{>\alpha}): \\
\displaystyle\frac{1}{\alpha K}\sum_{u = 1}^{\alpha K} \sqrbrkt{P_{\bfx|\bfu = u} V_{\bfs|\bfx,\bfu^{\le\alpha} = u}}_{\bfs} + \frac{1}{(1-\alpha)K}\sum_{u = \alpha K+1}^K\sqrbrkt{P_{\bfx|\bfu = u} ^{\ot2}V_{\bfs|\bfx,\bfx',\bfu^{>\alpha} = u}}_\bfs \in \lambda_\bfs
\end{array} }, \notag 
\end{align}
and
\begin{align}
\cF(P_{\bfx|\bfu}) &\coloneqq \curbrkt{ V_{\bfs|\bfx, \bfu}\in \Delta(\cS|\cX\times\cU): 
\frac{1}{K} \sum_{u = 1}^K \sqrbrkt{P_{\bfx|\bfu = u} V_{\bfs|\bfx, \bfu = u}}_\bfs \in \lambda_\bfs }. \notag 
\end{align}
\end{definition}

Generalizing the definition above for any $P_{\bfu}$ and slackness $\delta$ we have:
\begin{definition}[$\delta$-feasible jamming distributions]
Let $ (P_{\bfu}, P_{\bfx|\bfu})\in\Delta(\cU)\times\Delta(\cX|\cU) $. 
Define, for $ \alpha\in\cA $, 
\begin{align}
\cF_{\alpha,\delta}(P_\bfu, P_{\bfx|\bfu}) &\coloneqq \curbrkt{ \begin{array}{l}
(V_{\bfs|\bfx, \bfu^{\le\alpha}}, V_{\bfs|\bfx, \bfx', \bfu^{>\alpha}})\in \Delta(\cS|\cX\times\cU^{\le\alpha})\times\Delta(\cS|\cX^2\times\cU^{>\alpha}): \\
\displaystyle  \sum_{u \in \cU^{\le\alpha}} P_\bfu(u) \sqrbrkt{P_{\bfx|\bfu = u} V_{\bfs|\bfx,\bfu^{\le\alpha} = u}}_{\bfs} +  \sum_{u \in\cU^{>\alpha}} P_\bfu(u) \sqrbrkt{P_{\bfx|\bfu = u} ^{\ot2}V_{\bfs|\bfx,\bfx',\bfu^{>\alpha} = u}}_\bfs \in \interior_{\delta}(\lambda_\bfs)
\end{array} }, \notag 
\end{align}
and
\begin{align}
\cF_{\delta}(P_{\bfx|\bfu}) &\coloneqq \curbrkt{ V_{\bfs|\bfx, \bfu}\in \Delta(\cS|\cX\times\cU): 
 \sqrbrkt{P_\bfu P_{\bfx|\bfu } V_{\bfs|\bfx, \bfu }}_\bfs \in \interior_{\delta}(\lambda_\bfs) }. \notag 
\end{align}
\end{definition}

%\mikel{We should remove the duplicate of the definition above appearing in the Converse part}

\begin{definition}[Cumulative mutual information]
\label{def:cumulative-mutual-info}
% \yihan{Define it for general $ P_\bfu $.}
Fix $ P_{\bfx|\bfu} $, $ \alpha\in\cA $ and $ V_{\bfs|\bfx,\bfu^{\le\alpha}}\in\Delta(\cS|\cX\times\cU^{\le\alpha}) $. 
The \emph{cumulative mutual information} w.r.t.\ $P_{\bfx|\bfu}$ and $V_{\bfs|\bfx, \bfu^{\le\alpha}}$ is defined as
\begin{align}
I(P_{\bfx|\bfu}, V_{\bfs|\bfx, \bfu^{\le\alpha}}) &\coloneqq I(\bfx;\bfy|\bfu^{\le\alpha})
= \frac{1}{K} \sum_{u = 1}^{\alpha K} I(\bfx_u;\bfy_u), \notag 
\end{align}
where the joint distribution of $ (\bfx_u, \bfy_u) $ is given by
\begin{align}
P_{\bfx_u,\bfy_u}(x, y) &\coloneqq \sum_{s\in\cS} P_{\bfx|\bfu}(x|u) V_{\bfs|\bfx, \bfu^{\le\alpha}}(s|x,u) W_{\bfy|\bfx,\bfs}(y|x,s). \notag 
\end{align}
\end{definition}

Similarly, we generalize the above definition for any $ P_{\bfu^{\le\alpha}} $.

\begin{definition}[Cumulative mutual information]
\label{def:cumulative-mutual-info-converse}
Define the \emph{cumulative mutual information} w.r.t.\ $ (P_{\bfu^{\le\alpha}}, P_{\bfx|\bfu^{\le\alpha}}, V_{\bfs|\bfx,\bfu^{\le\alpha}})\in\Delta(\cU^{\le\alpha})\times\Delta(\cX|\cU^{\le\alpha})\times\Delta(\cS|\cX\times\cU^{\le\alpha}) $ as
\begin{align}
I(P_{\bfu^{\le\alpha}}, P_{\bfx|\bfu^{\le\alpha}}, V_{\bfs|\bfx,\bfu^{\le\alpha}}) &\coloneqq I(\bfx;\bfy|\bfu^{\le\alpha})
= \alpha \sum_{u\in\cU^{\le\alpha}} P_{\bfu^{\le\alpha}}(u) I(\bfx_u;\bfy_u), \notag 
\end{align} 
where the joint distribution of $ (\bfx_u,\bfy_u) $ is given by
\begin{align}
P_{\bfx_u,\bfy_u}(x,y) &\coloneqq \sum_{s\in\cS} P_{\bfx|\bfu^{\le\alpha}}(x|u) V_{\bfs|\bfx,\bfu^{\le\alpha}}(s|x,u) W_{\bfy|\bfx,\bfs}(y|x,s). \notag 
\end{align}
\end{definition}

\begin{remark}
In the above definition, the time-sharing distribution $ P_{\bfu^{\le\alpha}} $ is over the prefix $ \cU^{\le\alpha} $. 
However, the information throughput is measured w.r.t.\ the whole block of transmissions of length $n$. 
Therefore, we need an $\alpha$ factor to account for that. 
If we instead define the cumulative mutual information with $ P_\bfu\in\Delta(\cU) $, then no $\alpha$ factor is needed. 
\end{remark}

%\mikel{We may need to change the definition below slightly to take $\bfu$ into account and also to require a condition only for $P_{\bfx|\bfu = u}(x) \ne 0$.}

\begin{definition}[Symmetrizing distributions]\label{def:symm-dist}
Define
\begin{align}
\cV &\coloneqq \curbrkt{ \begin{array}{l}
V_{\bfs|\bfx,\bfx}\in\Delta(\cS|\cX^2): \\ 
\displaystyle\forall(x,x',y)\in\cX^2\times\cY,\, {\sum_{s\in\cS}V_{\bfs|\bfx,\bfx'}(s|x,x')W_{\bfy|\bfx,\bfs}(y|x,s) = \sum_{s\in\cS}V_{\bfs|\bfx,\bfx'}(s|x',x) W_{\bfy|\bfx,\bfs}(y|x',s)}
\end{array} } , \notag \\
\cV_\slackrate &\coloneqq \curbrkt{ \begin{array}{l}
		V_{\bfs|\bfx,\bfx}\in\Delta(\cS|\cX^2): \\ 
		\displaystyle\forall(x,x',y)\in\cX^2\times\cY,\, \left|{\sum_{s\in\cS}V_{\bfs|\bfx,\bfx'}(s|x,x')W_{\bfy|\bfx,\bfs}(y|x,s) -
			 \sum_{s\in\cS}V_{\bfs|\bfx,\bfx'}(s|x',x) W_{\bfy|\bfx,\bfs}(y|x',s)}\right| \leq \slackrate
\end{array} } , \notag \\
\cV' &\coloneqq \curbrkt{ \begin{array}{l}
(V_{\bfs|\bfx,\bfx}, V_{\bfs|\bfx,\bfx'}')\in\Delta(\cS|\cX^2)^2: \\ 
\displaystyle\forall(x,x',y)\in\cX^2\times\cY,\, {\sum_{s\in\cS}V_{\bfs|\bfx,\bfx'}(s|x,x')W_{\bfy|\bfx,\bfs}(y|x,s) = \sum_{s\in\cS}V_{\bfs|\bfx,\bfx'}'(s|x',x) W_{\bfy|\bfx,\bfs}(y|x',s)}
\end{array} } , \notag \\
{\cV}'_\delta &\coloneqq \curbrkt{ \begin{array}{l}
(V_{\bfs|\bfx,\bfx}, V_{\bfs|\bfx,\bfx'}')\in\Delta(\cS|\cX^2)^2: \\ 
\displaystyle\forall(x,x',y)\in\cX^2\times\cY,\, \abs{\sum_{s\in\cS}V_{\bfs|\bfx,\bfx'}(s|x,x')W_{\bfy|\bfx,\bfs}(y|x,s) - \sum_{s\in\cS}V_{\bfs|\bfx,\bfx'}'(s|x',x) W_{\bfy|\bfx,\bfs}(y|x',s)} \leq \delta
\end{array}  }. \notag 
\end{align}
\end{definition}

In what follows, we show that $C$ defined below is the causal-capacity for general AVC channels. We start by defining $C_K$ for every positive integer $K$.
Let $|\cU|=K$.
\begin{align}
	C_K \coloneqq \max_{\substack{P_{\bfx|\bfu}\in\Delta(\cX|\cU) \\ \sqrbrkt{P_\bfu P_{\bfx|\bfu}}_\bfx\in\lambda_\bfx}} 
	\min\Big\{
		& \min_{ V_{\bfs|\bfx,\bfu}\in\cF(P_{\bfx|\bfu})} I(P_{\bfx|\bfu}, V_{\bfs|\bfx,\bfu}) , \notag \\
		& \min_{\substack{(\alpha, (V_{\bfs|\bfx,\bfu^{\le\alpha}}, V_{\bfs|\bfx,\bfx',\bfu^{>\alpha}}))\in\cA\times\cF_{\alpha}(P_\bfu,P_{\bfx|\bfu}) \\
		% (V_{\bfs|\bfx,\bfu^{\le\alpha}}, V'_{\bfs|\bfx,\bfx',\bfu^{>\alpha}})\in \cF_{\alpha,\delta'}(P_{\bfx|\bfu})\\
		\forall u\in\cU^{>\alpha},\,  (V_{\bfs|\bfx,\bfx',\bfu^{>\alpha} = u}, V_{\bfs|\bfx,\bfx',\bfu^{>\alpha} = u}')\in\cV'}}
		I(P_{\bfx|\bfu^{\le\alpha}}, V_{\bfs|\bfx,\bfu^{\le\alpha}})\Big\}
		. \label{eqn:capacity} 
\end{align}

Averaging over $\frac{ V_{\bfs|\bfx,\bfx',\bfu^{>\alpha} = u}+ V_{\bfs|\bfx,\bfx',\bfu^{>\alpha} = u}'}{2}$ in the rightmost $\min$ in the expression above, we have

\begin{align}
	C_K \coloneqq  \max_{\substack{P_{\bfx|\bfu}\in\Delta(\cX|\cU) \\ \sqrbrkt{P_\bfu P_{\bfx|\bfu}}_\bfx\in\lambda_\bfx}} 
	\min\Big\{
		& \min_{ V_{\bfs|\bfx,\bfu}\in\cF(P_{\bfx|\bfu})} I(P_{\bfx|\bfu}, V_{\bfs|\bfx,\bfu}), \notag \\
		& \min_{\substack{(\alpha, (V_{\bfs|\bfx,\bfu^{\le\alpha}}, V_{\bfs|\bfx,\bfx',\bfu^{>\alpha}}))\in\cA\times\cF_{\alpha}(P_\bfu,P_{\bfx|\bfu}) \\ 
		% (V_{\bfs|\bfx,\bfu^{\le\alpha}}, V'_{\bfs|\bfx,\bfx',\bfu^{>\alpha}})\in \cF_{\alpha,\delta'}(P_{\bfx|\bfu})\\
		\forall u\in\cU^{>\alpha},\,  V_{\bfs|\bfx,\bfx',\bfu^{>\alpha} = u}\in\cV}}
		I(P_{\bfx|\bfu^{\le\alpha}}, V_{\bfs|\bfx,\bfu^{\le\alpha}})\Big\}. \label{eqn:converse-rate} 
\end{align}

Let $C = \limsup_{K \rightarrow \infty}C_K$.
% (we show that the limit exists in Claim~\ref{claim:upper_lower} below).

\begin{theorem}[Achievability]
\label{thm:achievability}
%\yihan{Statement of finite-length achievability. Need to specify the choice of $K$, $\gamma$ etc.\ as functions of $\slackrate$. }
For any $\slackrate_0>0$ rate $C-\slackrate_0$ is achievable.
\end{theorem}

\begin{theorem}[Converse]
	\label{thm:converse_main}
	For any $\slackrate_0>0$ rate $C+\slackrate_0$ is not achievable.
\end{theorem}

Before we prove the theorems above we state a technical claim.
For $\cU \coloneqq [K]$ and $\slackrate >0$, $\underline{C}_{K,\slackrate}$ defined below will be used in the proof of Theorem~\ref{thm:achievability}.
The remainder of this work includes several slackness parameters that are mutually related and highlighted in blue to ease the reading and verification process.
% = \sup{\cR_\slackrate}$. 
%Then, by studying the complement of $\cR_{\slackrate}$ we have that 
\begin{align}
	\underline{C}_{K,\slackrate} &\coloneqq \max_{\substack{P_{\bfx|\bfu}\in\Delta(\cX|\cU) \\ \sqrbrkt{P_\bfu P_{\bfx|\bfu}}_\bfx\in\lambda_\bfx}} 
	\min\curbrkt{\min_{ V_{\bfs|\bfx,\bfu}\in\cF(P_{\bfx|\bfu})} I(P_{\bfx|\bfu}, V_{\bfs|\bfx,\bfu}), 
	\min_{\substack{(\alpha, (V_{\bfs|\bfx,\bfu^{\le\alpha}}, V_{\bfs|\bfx,\bfx',\bfu^{>\alpha}}))\in\cA\times\cF_\alpha(P_{\bfx|\bfu}) \\ 
	(V_{\bfs|\bfx,\bfu^{\le\alpha}}, V'_{\bfs|\bfx,\bfx',\bfu^{>\alpha}})\in \cF_\alpha(P_{\bfx|\bfu})\\
	\forall u\in\cU^{>\alpha},\,  (V_{\bfs|\bfx,\bfx',\bfu^{>\alpha} = u}, V_{\bfs|\bfx,\bfx',\bfu^{>\alpha} = u}')\in\cV'_\slackrate}}
	I(P_{\bfx|\bfu}, V_{\bfs|\bfx,\bfu^{\le\alpha}})}. \notag 
\end{align}
Averaging over $\frac{ V_{\bfs|\bfx,\bfx',\bfu^{>\alpha} = u}+ V_{\bfs|\bfx,\bfx',\bfu^{>\alpha} = u}'}{2}$ in the rightmost $\min$ in the expression above, we have
\begin{align}
	\underline{C}_{K,\slackrate} &\coloneqq \max_{\substack{P_{\bfx|\bfu}\in\Delta(\cX|\cU) \\ \sqrbrkt{P_\bfu P_{\bfx|\bfu}}_\bfx\in\lambda_\bfx}} 
	\min\curbrkt{\min_{ V_{\bfs|\bfx,\bfu}\in\cF(P_{\bfx|\bfu})} I(P_{\bfx|\bfu}, V_{\bfs|\bfx,\bfu}), 
		\min_{\substack{(\alpha, (V_{\bfs|\bfx,\bfu^{\le\alpha}}, V_{\bfs|\bfx,\bfx',\bfu^{>\alpha}}))\in\cA\times\cF_\alpha(P_{\bfx|\bfu}) \\ 
				\forall u\in\cU^{>\alpha},\,  V_{\bfs|\bfx,\bfx',\bfu^{>\alpha} = u} \in \cV_\slackrate}}
		I(P_{\bfx|\bfu}, V_{\bfs|\bfx,\bfu^{\le\alpha}})}. \notag 
\end{align}

For $\cU \coloneqq [K]$, $\slackrate >0$, and $\slk{state}=\poly(\delta)$, $\bar{C}_{K,\slackrate}$ defined below will be used in the proof of Theorem~\ref{thm:converse_main}.
 
\begin{align}
	\bar{C}_{K,\slackrate} \coloneqq \max_{\substack{P_{\bfx|\bfu}\in\Delta(\cX|\cU) \\ \sqrbrkt{P_\bfu P_{\bfx|\bfu}}_\bfx\in\lambda_\bfx}} 
	\min\Bigg\{
		& \min_{ V_{\bfs|\bfx,\bfu}\in\cF_{\slk{state}}(P_{\bfx|\bfu})} I(P_{\bfx|\bfu}, V_{\bfs|\bfx,\bfu}), \notag \\
		& \min_{\substack{(\alpha, (V_{\bfs|\bfx,\bfu^{\le\alpha}}, V_{\bfs|\bfx,\bfx',\bfu^{>\alpha}}))\in\cA\times\cF_{\alpha,\slk{state}}(P_{\bfx|\bfu}) \\ 
		% (V_{\bfs|\bfx,\bfu^{\le\alpha}}, V'_{\bfs|\bfx,\bfx',\bfu^{>\alpha}})\in \cF_{\alpha,\delta'}(P_{\bfx|\bfu})\\
		\forall u\in\cU^{>\alpha},\,  V_{\bfs|\bfx,\bfx',\bfu^{>\alpha} = u}\in\cV}}
		I(P_{\bfx|\bfu^{\le\alpha}}, V_{\bfs|\bfx,\bfu^{\le\alpha}})
	\Bigg\}. \label{eqn:converse-rate} 
\end{align}

Averaging over $\frac{ V_{\bfs|\bfx,\bfx',\bfu^{>\alpha} = u}+ V_{\bfs|\bfx,\bfx',\bfu^{>\alpha} = u}'}{2}$ in the rightmost $\min$ in the expression above, we have
\begin{align}
\bar{C}_{K,\slackrate} \coloneqq \max_{\substack{P_{\bfx|\bfu}\in\Delta(\cX|\cU) \\ \sqrbrkt{P_\bfu P_{\bfx|\bfu}}_\bfx\in\lambda_\bfx}} 
	\min\Bigg\{
		& \min_{ V_{\bfs|\bfx,\bfu}\in\cF_{\slk{state}}(P_{\bfx|\bfu})} I(P_{\bfx|\bfu}, V_{\bfs|\bfx,\bfu}), \notag \\
		& \min_{\substack{(\alpha, (V_{\bfs|\bfx,\bfu^{\le\alpha}}, V_{\bfs|\bfx,\bfx',\bfu^{>\alpha}}))\in\cA\times\cF_{\alpha,\slk{state}}(P_{\bfx|\bfu}) \\ 
		% (V_{\bfs|\bfx,\bfu^{\le\alpha}}, V'_{\bfs|\bfx,\bfx',\bfu^{>\alpha}})\in \cF_{\alpha,\delta'}(P_{\bfx|\bfu})\\
		\forall u\in\cU^{>\alpha},\,  V_{\bfs|\bfx,\bfx',\bfu^{>\alpha} = u} \in \cV}}
		I(P_{\bfx|\bfu^{\le\alpha}}, V_{\bfs|\bfx,\bfu^{\le\alpha}})
	\Bigg\}.
\end{align}
%\subsection{For discussion towards ISIT}
%
%The bad case is that there exists $\cU, (P_\bfu, P_{\bfx|\bfu})\in\Delta(\cU)\times\Delta(\cX|\cU)$,  $\sqrbrkt{P_{\bfx|\bfu}}_\bfx\in\lambda_\bfx$ such that the set 
%$$
%\cF_{\alpha}(P_\bfu,P_{\bfx|\bfu}) \cap \cV \subseteq \partial(\cF_{\alpha}(P_\bfu,P_{\bfx|\bfu}))
%$$
%

\begin{claim}
\label{claim:upper_lower}
%Let $\slackrate >0$. 
%Let $\slk{K}$ be a function of $\slackrate$ to be defined later.
%Let $C=\sup_{\slackrate,K=\slk{K}}C_{K,\slackrate} $. In what follows, we that $C$ is the causal-capacity for general AVC channels. Specifically, we show that
%\end{claim}
%
%\begin{claim}
Let $K$ be a positive integer.
Let $\delta$ be sufficiently small and let $\slk{state}=\poly(\delta)$, it holds that
$|C_K-\bar{C}_{K,\delta}| \leq \slk{GAP}$ and $|C_K-\underline{C}_{K,\delta}| \leq \slk{GAP}$, where $\slk{GAP}$ depends only on $\poly(\delta)$ , the parameters of the channel $W$, and the cost function $\cost$.
\end{claim}

\proof
%We show that for sufficiently large $K$, $\delta$ sufficiently small, and $\slk{state}=\poly(\delta)$ that
%$\bar{C}_{K,\delta} - \underline{C}_{K,\delta} \leq \slk{GAP}$ where $\slk{\delta}=\poly(\delta)$ for a suitable function that depends polynomially on $\delta$, $K$, the parameters of the channel $W$, and the cost function $\cost$.

Let $\delta>0$ be a sufficiently small slackness parameter (the requirements on $\delta$ will be discussed below).
Let $|\cU|=K$.
Let $P_{\bfx|\bfu}\in\Delta(\cX|\cU)$ come $\delta$-close to optimizing $C_K$ as defined in Equation~(\ref{eqn:capacity}). We compare the first and second expressions in the definition of $C_K$ with those of $\bar{C}_{K,\delta}$ and $\underline{C}_{K,\delta}$ when evaluated with $\cU$ and
$P_{\bfx|\bfu}\in\Delta(\cX|\cU)$. 
 By our definitions, it follows that $C_K \leq \bar{C}_{K,\delta}$ and $C_K \geq \underline{C}_{K,\delta}$.
 Thus to prove our assertion, it suffices to compare the first and second expressions in $\bar{C}_{K,\delta}$ and $\underline{C}_{K,\delta}$ to bound $\bar{C}_{K,\delta} - \underline{C}_{K,\delta}$, and thus, in turn, to bound $|C_K-\bar{C}_{K,\delta}|$ and $|C_K-\underline{C}_{K,\delta}|$.
 
 We note that once $P_{\bfx|\bfu}$ is fixed, the sets $\cF$, $\cF_{\alpha}$, and $V$ are convex sets defined by linear constraints. 
 Similarly for the ``$\slk{state}$-contracted'' versions $\cF_{\slk{state}}$ and $\cF_{\alpha,\slk{state}}$ or the ``$\delta$-extracted'' variant $V_\delta$.
 We also note that by the properties given in Section~\ref{sec:global-assumptions}, for sufficiently small $\delta$, we have that   $\cF_{\slk{state}}$ and $\cF_{\alpha,\slk{state}}$ are non-empty.
 
One can bound the difference between the first expression in $\bar{C}_{K,\delta}$ and $\underline{C}_{K,\delta}$ using our observation that  $\cF_{\slk{state}}$ and $\cF_{\alpha,\slk{state}}$ are non-empty combined with standard bounds on the mutual-information function.
More formally, for any optimizing $\underline{V}_{\bfs|\bfx,\bfu} \in \cF$ for $\underline{C}_{K,\delta}$ one can find $\bar{V}_{\bfs|\bfx,\bfu}\in \cF_{\slk{state}}$ for $\bar{C}_{K,\delta}$ for which $\|\underline{V}_{\bfs|\bfx,\bfu}- \bar{V}_{\bfs|\bfx,\bfu}\|_{\infty}$ is at most a linear function of $\slk{state}$ and other constant parameters involved in the definition of $\cF$.

For the second expression, we study the set $\cV$ and $\cV_\delta$ in combination with $\cF_{\alpha}$ and $\cF_{\alpha,\slk{state}}$. Namely, we study the set $\cV \cap [\cF_{\alpha,\slk{state}}]_{V_{\bfs|\bfx,\bfx',\bfu^{>{\alpha}}}}$ (where the latter set $\cF_{\alpha,\slk{state}}$ in the intersection is projected onto the variables of $V_{\bfs|\bfx,\bfx',\bfu^{>{\alpha}}}$). 
If this intersection is non-empty for any sufficiently small $\delta$, then so is $\cV_\delta \cap [\cF_{\alpha}]_{V_{\bfs|\bfx,\bfx',\bfu^{>{\alpha}}}}$, and thus, as previously discussed in the analysis of the first expression, 
for any optimizing $({\alpha}, (\underline{V}_{\bfs|\bfx,\bfu^{\le{\alpha}}}, \underline{V}_{\bfs|\bfx,\bfx',\bfu^{>{\alpha}}}))\in\cA\times\cF_{{\alpha}}$ such that  $\forall u\in\cU^{>{\alpha}},\,  \underline{V}_{\bfs|\bfx,\bfx',\bfu = u} \in \cV_\slackrate$ in $\underline{C}_{K,\delta}$ one can find a corresponding 
$(\alpha, (\bar{V}_{\bfs|\bfx,\bfu^{\le\alpha}}, \bar{V}_{\bfs|\bfx,\bfx',\bfu^{>\alpha}}))\in\cA\times\cF_{\alpha,\slk{state}}$ such that  $\forall u\in\cU^{>\alpha},\,  \bar{V}_{\bfs|\bfx,\bfx',\bfu = u} \in \cV$ in $\bar{C}_{K,\delta}$ for which 
$\|(\underline{V}_{\bfs|\bfx,\bfu^{\le{\alpha}}}, \underline{V}_{\bfs|\bfx,\bfx',\bfu^{>{\alpha}}})-  (\bar{V}_{\bfs|\bfx,\bfu^{\le\alpha}}, \bar{V}_{\bfs|\bfx,\bfx',\bfu^{>\alpha}})\|_{\infty}$ is at most a linear function of $\delta$, $\slk{state}$ and other constant parameters involved in the definitions of $\cF$ and $\cV$.

If, alternatively, for some small $\delta$ the intersection $\cV_\delta \cap [\cF_{\alpha}]_{V_{\bfs|\bfx,\bfx',\bfu^{>{\alpha}}}}$ is empty, then so is $\cV \cap [\cF_{\alpha,\slk{state}}]_{V_{\bfs|\bfx,\bfx',\bfu^{>{\alpha}}}}$ and $\bar{C}_{K,\delta}$ and $\underline{C}_{K,\delta}$  are not impacted by the second expressions under study.

Finally, if for any $\delta>0$, $\cV_\delta \cap [\cF_{\alpha}(P_{\bfx|\bfu})]_{V_{\bfs|\bfx,\bfx',\bfu^{>{\alpha}}}}$ is non-empty but 
$\cV \cap [\cF_{\alpha,\slk{state}}(P_{\bfx|\bfu})]_{V_{\bfs|\bfx,\bfx',\bfu^{>{\alpha}}}}$ is empty, we slightly modify the studied $P_{\bfx|\bfu}$ to $P'_{\bfx|\bfu}$ to guarantee that the new sets $\cV_\delta \cap [\cF_{\alpha}(P'_{\bfx|\bfu})]_{V_{\bfs|\bfx,\bfx',\bfu^{>{\alpha}}}}$ and 
$\cV \cap [\cF_{\alpha,\slk{state}}(P'_{\bfx|\bfu})]_{V_{\bfs|\bfx,\bfx',\bfu^{>{\alpha}}}}$ are both non-empty, which brings us to a case analyzed above.
Our modification on $P_{\bfx|\bfu}$ will be of limited statistical distance implying a corresponding limited impact on the objective $I(P_{\bfx|\bfu^{\le\alpha}}, V_{\bfs|\bfx,\bfu^{\le\alpha}})$.
To modify $P_{\bfx|\bfu}$, we fix an element $\bar{V}_{\bfs|\bfx,\bfx',\bfu^{>{\alpha}}} \in \cV \cap [\cF_{\alpha}(P_{\bfx|\bfu})]_{V_{\bfs|\bfx,\bfx',\bfu^{>{\alpha}}}}$ and its corresponding $\bar{V}_{\bfs|\bfx,\bfu^{\le\alpha}}$, the existence of which follows from our assumption that for any $\delta>0$, $\cV_\delta \cap [\cF_{\alpha}(P_{\bfx|\bfu})]_{V_{\bfs|\bfx,\bfx',\bfu^{>{\alpha}}}}$ is non-empty.
With $(\bar{V}_{\bfs|\bfx,\bfu^{\le\alpha}},\bar{V}_{\bfs|\bfx,\bfx',\bfu^{>{\alpha}}})$ in mind, $\cF_\alpha$ is a structured quadratic function of $P_{\bfx|\bfu}$.
One can now reduce the state-cost of $(P'_{\bfx|\bfu}, (\bar{V}_{\bfs|\bfx,\bfu^{\le\alpha}},\bar{V}_{\bfs|\bfx,\bfx',\bfu^{>{\alpha}}}))$ when compared to $(P_{\bfx|\bfu}, (\bar{V}_{\bfs|\bfx,\bfu^{\le\alpha}},\bar{V}_{\bfs|\bfx,\bfx',\bfu^{>{\alpha}}}))$ by $\slk{state}$ by modifying $P_{\bfx|\bfu}$ to $P'_{\bfx|\bfu}$ with 
$\|P_{\bfx|\bfu}-P'_{\bfx|\bfu}\|_\infty$ bounded by a polynomial in $\slk{state}$ and other constant parameters involved in the definition of $\cF_\alpha$.
Accordingly, the objective $I(P'_{\bfx|\bfu},\bar{V}_{\bfs|\bfx,\bfu^{\le\alpha}})$ when compared to $I(P_{\bfx|\bfu},\bar{V}_{\bfs|\bfx,\bfu^{\le\alpha}})$ is similarly bounded.
This now implies that both $\cV_\delta \cap [\cF_{\alpha}(P'_{\bfx|\bfu})]_{V_{\bfs|\bfx,\bfx',\bfu^{>{\alpha}}}}$ and $\cV \cap [\cF_{\alpha,\slk{state}}(P'_{\bfx|\bfu})]_{V_{\bfs|\bfx,\bfx',\bfu^{>{\alpha}}}}$ are non-empty, and we may use the analysis of the similar case studied above.

\section{Proof of achievability (Theorem~\ref{thm:achievability})}
\label{sec:proof-achievability}

% !TEX root = causal_general.tex

%Let the rate be $R = C - \slackrate$. In the proof that follows we use several slackness parameters defined as follws:
%\begin{align}
%\gamma &< \slk{LB}/2 
%	& \textrm{common randomness (Claim 6)} \\
%K &= \textcolor{red}{\epsilon^3}
%	& \textrm{number of chunks (Claim. 7)} \\	
%\eps &= \frac{n}{K}
%	& \textrm{fractional length of chunk} \\
%\slk{TYP} &\to 0 
%	& \textrm{typicality slack} \\
%\slk{x|u} &= 3 \eps
%	& \textrm{joint type of codewords (Claim 7)} \\
%\slk{UB} & 
%	& \textrm{upper bound for rate goodness}\\
%\slk{LB} & 
%	& \textrm{lower bound for rate goodness}\\
%\slk{GOOD} & 
%	& \textrm{goodness slack?}
%\slk{SYM} & 
%	& \textrm{symmetrizability slack}
%\end{align}
%Error probabilities:
%	\begin{align}
%	\pr{1} &< 2^{-n \gamma/4}
%		& \textrm{prob of list code being bad (construction, Claim 6)} \\
%	\pr{2} &< \exp\left( -\frac{3}{2} n \eps^3 \right)
%		& \textrm{prob of list code being bad (construction, Claim 7)}
%	\pr{3} &< 2^{- n \eps^3/3 }
%		& \textrm{prob of bad (msg,key) in the list (decoding, Claim 7)}
%	\end{align}
%
\proof

%Fix $ P_{\bfx|\bfu}\in\Delta(\cX|\cU) $. 
Let $ \delta_0 >0$.
Let $R = C-\delta_0$.
Let $\delta$ be a sufficiently small function of $\delta_0$ to be determined shortly.
Let $\slk{MIN}$ be a function of $\delta$ defined shortly through 
Claim~\ref{claim:min}.
Let $\varepsilon(\delta)=\delta\slk{MIN}^2/2$.
Let $K= 1/\varepsilon$.
%That is, $\varepsilon = \Theta(\delta^3)$.
Take $\delta$ to be sufficiently small such that $|C-C_k| \leq \delta_0/4$ and, by Claim~\ref{claim:upper_lower}, $|C_K-\underline{C}_{K,\delta}| \leq \delta_0/4$, implying that $R \leq \underline{C}_{K,\delta}-\delta$.

Let $\slk{UB} \coloneqq \varepsilon$, $\slk{LB} \coloneqq \slk{UB}+\varepsilon\log{|\cX|}$, and $\slk{x|u} \coloneqq 3\varepsilon$.
Note, with this setting of parameters, that $2\slk{x|u}/\slk{MIN}^2=\delta$ (this fact will be used later in our proof).
Without loss of generality, assume $ n $ is sufficiently large so that $ n\eps $ is an integer. 
Let $ \gamma \coloneqq \varepsilon^4$. 
%Let $ R\ge0 $ be $\slackrate$-good w.r.t.\ $ (P_\bfu,P_{\bfx|\bfu}) $ (as per \Cref{def:goodness-ach}). 
Let $ M \coloneqq 2^{nR}, N \coloneqq 2^{n\gamma} $ and $ \cM\coloneqq[M],\cR\coloneqq[N] $. 

Let 
\begin{align}
\vu &\coloneqq (\underbracket{1,\cdots,1}_{n\eps},\underbracket{2,\cdots,2}_{n\eps}, \cdots, \underbracket{K,\cdots,K}_{n\eps}) \in \cU^n. \notag 
\end{align}
For a vector $ \vv\in\Sigma^n $ over some alphabet $ \Sigma $ and $ \alpha\in\cA $, let
\begin{align}
\vv^{\le\alpha} &\coloneqq (\vv(1),\cdots,\vv(\alpha n)) \in\Sigma^{\alpha n}, \notag \\
\vv^{>\alpha} &\coloneqq (\vv(\alpha n + 1), \cdots, \vv(n)) \in\Sigma^{(1-\alpha)n}; \notag 
\end{align}
for $ u\in\cU $, let
\begin{align}
\vv^{(u)} &\coloneqq (\vv((u-1)n\eps+1),\vv((u-1) n\eps + 2),\cdots, \vv(un\eps)) \in \Sigma^{n\eps}. \notag 
\end{align}
For two vectors $ \vv_1\in\Sigma^{n_1} $ and $ \vv_2\in\Sigma^{n_2} $, denote their concatenation by 
\begin{align}
\vv_1\circ\vv_2 &\coloneqq (\vv_1(1),\cdots,\vv_1(n_1),\vv_2(1),\cdots,\vv_2(n_2)) \in\Sigma^{n_1+n_2}. \notag 
\end{align}
For two sets of vectors $ \cV_1\in\Sigma^{n_1} $ and $ \cV_2\in\Sigma^{n_2} $, let 
\begin{align}
\cV_1\circ\cV_2 &\coloneqq \curbrkt{\vv_1\circ\vv_2:\vv_1\in\cV_1,\vv_2\in\cV_2}. \notag 
\end{align}

Our proof uses the following notion of {\em code-goodness}:
\begin{definition}[Rate goodness]
\label{def:goodness-ach}
Let $\slackrate >0$. Fix $ P_{\bfx|\bfu}\in\Delta(\cX|\cU) $. 
We say that a rate $ R\ge0 $ is \emph{$\slackrate$-good} w.r.t.\ $ (P_\bfu, P_{\bfx|\bfu}) $ if 
\begin{enumerate}
	\item \begin{align} 
	R &\le \min_{ V_{\bfs|\bfx,\bfu}\in\cF(P_{\bfx|\bfu})} I(P_{\bfx|\bfu}, V_{\bfs|\bfx,\bfu}); \notag 
	\end{align}

	\item For any $ (\alpha, (V_{\bfs|\bfx,\bfu^{\le\alpha}}, V_{\bfs|\bfx,\bfx',\bfu^{>\alpha}}))\in\cA\times\cF_\alpha(P_{\bfx|\bfu}) $ satisfying $\alpha \leq 1-1/K$ and $ R \in [I(P_{\bfx|\bfu}, V_{\bfs|\bfx,\bfu^{\le\alpha}}) - \slk{LB}, [I(P_{\bfx|\bfu}, V_{\bfs|\bfx,\bfu^{\le\alpha}}) - \slk{UB}]$,
	the following property holds.
	\begin{itemize}

		\item For every $ V_{\bfs|\bfx,\bfx',\bfu^{>\alpha}}'\in\Delta(\cS|\cX^2\times\cU^{>\alpha}) $ such that $ (V_{\bfs|\bfx,\bfu^{\le\alpha}}, V_{\bfs|\bfx,\bfx',\bfu^{>\alpha}}')\in\cF_\alpha(P_{\bfx|\bfu}) $, there exists $ u\in\cU^{>\alpha} $ such that $ (V_{\bfs|\bfx,\bfx',\bfu^{>\alpha} = u}, V_{\bfs|\bfx,\bfx',\bfu^{>\alpha} = u}')\not\in\cV'_\slackrate $. 
	\end{itemize}
\end{enumerate}
\end{definition}

\subsection{Determining a  $P_{\bfx|\bfu}$ for code design}

We first show using Claim~\ref{claim:min} below that there exists a $P_{\bfx|\bfu} \in \Delta(\cX|\cU)$ with $\sqrbrkt{P_\bfu P_{\bfx|\bfu}}_\bfx\in\lambda_\bfx$, such that for any $u \in \cU$, it holds that 
$$
\min_{x; P_{\bfx|\bfu=u}(x)>0}\{P_{\bfx|\bfu=u}(x)\} \geq \slk{MIN}
$$ 
and in addition
\begin{align}
\label{eq:ck}
\underline{C}_{K,\delta} - \slackrate/2 \leq \min\curbrkt{\min_{ V_{\bfs|\bfx,\bfu}\in\cF(P_{\bfx|\bfu})} I(P_{\bfx|\bfu}, V_{\bfs|\bfx,\bfu}), 
	\min_{\substack{(\alpha, (V_{\bfs|\bfx,\bfu^{\le\alpha}}, V_{\bfs|\bfx,\bfx',\bfu^{>\alpha}}))\in\cA\times\cF_\alpha(P_{\bfx|\bfu}) \\ 
	(V_{\bfs|\bfx,\bfu^{\le\alpha}}, V'_{\bfs|\bfx,\bfx',\bfu^{>\alpha}})\in \cF_\alpha(P_{\bfx|\bfu})\\
	\forall u\in\cU^{>\alpha},\,  (V_{\bfs|\bfx,\bfx',\bfu^{>\alpha} = u}, V_{\bfs|\bfx,\bfx',\bfu^{>\alpha} = u}')\in\cV'_\slackrate}}
	I(P_{\bfx|\bfu}, V_{\bfs|\bfx,\bfu^{\le\alpha}})}.
\end{align}

\begin{claim}
\label{claim:min}
Let $\delta>0$. For a suitable $\slk{MIN}$ that depends polynomially on $\delta$ and additional parameters of the channel $W$ at hand it holds 
that for any  $P_{\bfx|\bfu} \in \Delta(\cX|\cU)$ with $\sqrbrkt{P_\bfu P_{\bfx|\bfu}}_\bfx\in\lambda_\bfx$ there exists $P'_{\bfx|\bfu} \in \Delta(\cX|\cU)$ with $\sqrbrkt{P_\bfu P'_{\bfx|\bfu}}_\bfx\in\lambda_\bfx$ such that 
\begin{enumerate}
\item $\min_{x; P'_{\bfx|\bfu=u}(x)>0}\{P'_{\bfx|\bfu=u}(x)\} \geq \slk{MIN}$.
\item The right-hand-side of  (\ref{eq:ck}) when evaluated on $P_{\bfx|\bfu} \in \Delta(\cX|\cU)$ or   $P'_{\bfx|\bfu} \in \Delta(\cX|\cU)$ differs by at most $\delta/2$.
%\item For any $(V_{\bfs|\bfx,\bfu^{\le\alpha}}, V_{\bfs|\bfx,\bfx',\bfu^{>\alpha}})$ and $\alpha$, $(V_{\bfs|\bfx,\bfu^{\le\alpha}}, V_{\bfs|\bfx,\bfx',\bfu^{>\alpha}})\in \cF_\alpha(P'_{\bfx|\bfu})$ implies $(V_{\bfs|\bfx,\bfu^{\le\alpha}}, V_{\bfs|\bfx,\bfx',\bfu^{>\alpha}})\in \cF_\alpha(P_{\bfx|\bfu})$.
%\item $\sqrbrkt{P_\bfu P_{\bfx|\bfu}}_\bfx\in\lambda_\bfx$ implies $\sqrbrkt{P_\bfu P'_{\bfx|\bfu}}_\bfx\in\lambda_\bfx$.
%\item For any $(V_{\bfs|\bfx,\bfu^{\le\alpha}}, V_{\bfs|\bfx,\bfx',\bfu^{>\alpha}})$ and $\alpha$, $|I(P_{\bfx|\bfu}, V_{\bfs|\bfx,\bfu^{\le\alpha}})-I(P'_{\bfx|\bfu}, V_{\bfs|\bfx,\bfu^{\le\alpha}})| \leq \delta/2$.
\end{enumerate}
\end{claim}

\proof
Follows lines similar to those presented in the proof of Claim~\ref{claim:upper_lower}. 
%\mikel{Need to double check bounds above.}

\vspace{2mm}

\noindent
Fix $ P_{\bfx|\bfu}\in\Delta(\cX|\cU) $ to satisfy the above.
We now show that $R$ is \emph{$\slackrate$-good} w.r.t.\ $ (P_\bfu, P_{\bfx|\bfu}) $.
Namely, by our definition of $ P_{\bfx|\bfu}\in\Delta(\cX|\cU) $, 
$R \leq \underline{C}_{K,\delta} -\delta \le \min_{ V_{\bfs|\bfx,\bfu}\in\cF(P_{\bfx|\bfu})} I(P_{\bfx|\bfu}, V_{\bfs|\bfx,\bfu})$.
Moreover, let $ (\alpha, (V_{\bfs|\bfx,\bfu^{\le\alpha}}, V_{\bfs|\bfx,\bfx',\bfu^{>\alpha}}))\in\cA\times\cF_\alpha(P_{\bfx|\bfu}) $ satisfy $\alpha \leq 1-1/K$ and $ R \in [I(P_{\bfx|\bfu}, V_{\bfs|\bfx,\bfu^{\le\alpha}}) - \slk{LB}, [I(P_{\bfx|\bfu}, V_{\bfs|\bfx,\bfu^{\le\alpha}}) - \slk{UB}]$, then
for every $ V_{\bfs|\bfx,\bfx',\bfu^{>\alpha}}'\in\Delta(\cS|\cX^2\times\cU^{>\alpha}) $ such that $ (V_{\bfs|\bfx,\bfu^{\le\alpha}}, V_{\bfs|\bfx,\bfx',\bfu^{>\alpha}}')\in\cF_\alpha(P_{\bfx|\bfu}) $, there exists $ u\in\cU^{>\alpha} $ such that $ (V_{\bfs|\bfx,\bfx',\bfu^{>\alpha} = u}, V_{\bfs|\bfx,\bfx',\bfu^{>\alpha} = u}')\not\in\cV'_\slackrate $. 
As otherwise, the right-hand-side of (\ref{eq:ck}) would be less than $R+\slk{LB} =  \underline{C}_{K,\delta}-\delta+\slk{LB} < \underline{C}_{K,\delta}-\delta/2$ in contradiction to  (\ref{eq:ck}).

%
%
%\subsection{Notation}
%\label{sec:notation-ach-pf}

\subsection{Code construction}
\label{sec:code-construction}
We describe below the construction of the codebook $ \cC\subset\cX^n $. 
It is a concatenation of $ K $ short codebooks: $ \cC = \cC^{(1)}\circ\cC^{(2)}\circ\cdots\circ\cC^{(K)} $. 
Here for each $ u\in\cU $, $ \cC^{(u)}\subset\cX^{n\eps} $ is sampled randomly and independently. 
Specifically, $ |\cC^{(u)}| = M\cdot N $ and each $ \vx^{(u)}\in\cC^{(u)} $ is independent and uniformly distributed in $ \cT_{n\eps}(P_{\bfx|\bfu = u}) $. 
We label codewords in $ \cC^{(u)} $ using a pair $ (m,r)\in\cM\times\cR $, i.e., $ \cC^{(u)} = \{\vx^{(u)}_{(m,r)}\}_{(m,r)\in\cM\times\cR} $.

\subsection{Encoding}
\label{sec:enc}
The encoder $\enc$ is stochastic. 
To encode a message $ m\in\cM $, the encoder first samples a sequence of random seeds $ \bfr = (\bfr_1,\cdots,\bfr_K)\in\cR^K $, where $ \bfr_u\iid\unif(\cR) $ for each $ u\in\cU $. 
These seeds are private to the encoder and not revealed to any other parties. 
The encoding of $m$ with a random seed $\bfr$ is given by $ \enc(m,\bfr) \coloneqq \vx_{(m,\bfr_1)}^{(1)}\circ\vx_{(m,\bfr_2)}^{(2)}\circ\cdots\circ\vx_{(m,\bfr_K)}^{(K)} $, where we make the dependence on $\bfr$ explicit. 

\subsection{Decoding}
\label{sec:dec}

Bob's decoder runs iteratively. 
Suppose the channel output is $ \vy\in\cY^n $. 
For each $ \alpha \in\cA $ (from small to large values, i.e., $ \alpha = 0,1/n,2/n,\cdots,1 $) and for each $V_{\bfs|\bfx,\bfu^{\le\alpha}}$ for which there exists a 
$V_{\bfs|\bfx,\bfx',\bfu^{>\alpha}}$ such that 
$(V_{\bfs|\bfx,\bfu^{\le\alpha}},V_{\bfs|\bfx,\bfx',\bfu^{>\alpha}})\in\cF_\alpha(P_{\bfx|\bfu}) $ satisfying $ R \in  [I(P_{\bfx|\bfu},V_{\bfs|\bfx,\bfu^{\le\alpha}}) - \slk{LB}, I(P_{\bfx|\bfu},V_{\bfs|\bfx,\bfu^{\le\alpha}})-\slk{UB}] $, the decoder performs the following two-step decoding. 

\begin{enumerate}
	\item \textbf{List-decoding.} 
	The decoder first list-decodes the prefix $ (\vy(1),\cdots,\vy(\alpha K)) $ to the following list $ \cM(P_{\bfx|\bfu}, V_{\bfs|\bfx,\bfu^{\le\alpha}}) $:
	\begin{align}
	 \cM(P_{\bfx|\bfu}, V_{\bfs|\bfx,\bfu^{\le\alpha}}) &\coloneqq \curbrkt{ m\in\cM:\begin{array}{l}
	 \exists(r_1,\cdots,r_{\alpha K})\in\cR^{\alpha K},\exists(\vs^{(1)},\cdots,\vs^{(\alpha K)})\in(\cS^{n\eps})^{\alpha K} \,\suchthat \\
	 \forall u\in\cU^{\le\alpha},\,\type_{\vx_{(m,r_u)}^{(u)},\vs^{(u)},\vy^{(u)}} 
	 %\in \Delta_{\slk{LIST}}(
	 = P_{\bfx|\bfu = u} V_{\bfs|\bfx,\bfu^{\le\alpha} = u} W_{\bfy|\bfx,\bfs}
	 \end{array}}. \notag 
	\end{align}
    %Here, for a distribution $P$, $\Delta_\slackrate(P)$ includes all distributions within statistical distance $\slackrate$ from $P$.
	In words, the list consists of all messages whose corresponding codewords are 
	%(approximately) 
	jointly typical with the prefix of the received vector according to the joint distribution $ P_{\bfu}P_{\bfx|\bfu}V_{\bfs|\bfx,\bfu^{\le\alpha}}W_{\bfy|\bfx,\bfs} $. 
	We also define $ \cL(P_{\bfx|\bfu},V_{\bfs|\bfx,\bfu^{\le\alpha}}) $ as the set of codewords whose corresponding messages are in $ \cM(P_{\bfx|\bfu},V_{\bfs|\bfx,\bfu^{\le\alpha}}) $.
%	\begin{align}
%	\cL(P_{\bfx|\bfu},V_{\bfs|\bfx,\bfu^{\le\alpha}}) &\coloneqq \curbrkt{\vx\in\cC:
%	% \begin{array}{l}
%	 % \exists(r_1,\cdots,r_{\alpha K})\in\cR^{\alpha K},
%	 \exists(\vs^{(1)},\cdots,\vs^{(\alpha K)})\in(\cS^{n\eps})^{\alpha K}, \, %\\
%	 \forall u\in\cU^{\le\alpha},\,\type_{\vx^{(u)},\vs^{(u)},\vy^{(u)}} = P_{\bfx|\bfu = u} V_{\bfs|\bfx,\bfu^{\le\alpha} = u} W_{\bfy|\bfx,\bfs}
%	 % \end{array}
%	 }. \notag 
%	\end{align}

	\item \textbf{Unique decoding.} 
	Bob examines each codeword $\vx' \in 	\cL(P_{\bfx|\bfu},V_{\bfs|\bfx,\bfu^{\le\alpha}})$
	and outputs the message $m'$ corresponding to $\vx'$ if and only if $\vy$ could have been obtained from $\vx'$ through a feasible $\vs'$. In particular, if there exists a vector $\vs'$ and  $V_{\bfs|\bfx,\bfu^{>\alpha}}$ for which $(V_{\bfs|\bfx,\bfu^{\le\alpha}},V_{\bfs|\bfx,\bfu^{>\alpha}}) \in \cF_\alpha(P_{\bfx|\bfu})$ and $\forall u\in\cU^{>\alpha}$, $\type_{\vx'^{(u)},\vs'^{(u)},\vy^{(u)}} 
	%\in \Delta_{\slk{LIST}} (
	= P_{\bfx|\bfu = u} V_{\bfs|\bfx,\bfu^{>\alpha} = u} W_{\bfy|\bfx,\bfs}$.
	I.e., corrupting $\vx'$ using $\vs'$ results in the received word $\vy$.
	If no $\vx' \in 	\cL(P_{\bfx|\bfu},V_{\bfs|\bfx,\bfu^{\le\alpha}})$ passes the test above, Bob continues in studying the next pair $(\alpha,V_{\bfs|\bfx,\bfu^{\le\alpha}})$ until eventually finding a codeword and a corresponding message that pass the test. Once such a codeword is found the decoding process is terminated. If no codewords pass the test in any of Bob's iterations, a decoding error is considered.
\end{enumerate}

\noindent
%\mikel{
%	\noindent
%	Mike's remarks
%	\begin{enumerate}
%		\item We probably need to quantize the set of $V_{\bfs|\bfx,\bfu^{\le\alpha}}$ studied, as this will appear in our union bound. That way, the number of different potential lists is bounded by the code size (the transmitted codeword) times $\cS^n$ times the number of different $(\alpha,V_{\bfs|\bfx,\bfu^{\le\alpha}})$.
%	\end{enumerate}
%}
	
\subsection{Analysis}
We now show that with high probability over code design it is the case that for every transmitted message $m$ with high probability over the stochasticity of Alice the decoding process will  succeed. Namely, with high probability over the stochasticity of Alice, only the the codeword $\vx_{(m,r)} $ corresponding to $m$ will pass the unique decoding step.

\subsubsection{Code properties}
We start by analyzing some properties of our code:

% !TEX root = causal_general.tex

\begin{claim}[List size]
\label{claim:listsize}
With probability at most $\pr{1} = 2^{-n \gamma/4}$ over the code design, if we choose $\gamma < \slk{UB}/2$, for any pair $(\alpha,V_{\bfs|\bfx,\bfu^{\le\alpha}})$, the size of the list produced in the first step of Bob's decoding exceeds $\listsize = \frac{ 2 \log |\cY| }{ \slk{UB} }$.
\end{claim}

%\mikel{A connection between $\slk{UB}$ and $\listsize$ we be determined in analysis above.}

\begin{IEEEproof}
%\mikel{Unless I am missing something, we are only decoding according to a given $V_{\bfs|\bfx, \bfu^{\le\alpha}}$, so I omitted the definition of $R_0$ and other corresponding points below.}
Fix a pair $(\alpha,V_{\bfs|\bfx,\bfu^{\le\alpha}})$ such that $(V_{\bfs|\bfx,\bfu^{\le\alpha}},V_{\bfs|\bfx,\bfx',\bfu^{>\alpha}})\in\cF_\alpha(P_{\bfx|\bfu})$. 
%Let $R_0 = I(P_{\bfx|\bfu}, V_{\bfs|\bfx, \bfu^{\le\alpha}})$.%Consider any output prefix $\vy^{\le \alpha} = (\vy(1), \vy(2), \ldots, \vy(\alpha n))$. 
For each $u  \in \cU$,
%we have a single letter channel $V_{\bfs|\bfx,\bfu=u}(s|x,u)$ and 
define the joint distribution $P_{\bfx, \bfy|\bfu=u} = [P_{\bfx|\bfu=u} V_{\bfs|\bfx,\bfu = u}W_{\bfy|\bfx,\bfs}]_{\bfx, \bfy|\bfu=u}$ and the ``reverse channel'' $V_{\bfx|\bfy, \bfu = u}$ by conditioning $P_{\bfx, \bfy|\bfu=u}$.
Consider any output prefix $\vy^{\le \alpha} = (\vy(1), \vy(2), \ldots, \vy(\alpha K))$ where $\vy(u)$ corresponds to chunk $u$ of $\vy^{\le \alpha}$.
For each $u$, let $\cA_{u}(\vy(u)) = \{\vx \in \cX^{n\eps}: T_{\vx} = V_{\bfx|\bfy = T_{\vy(u)}, \bfu = u}\}$ be the set of $\vx \in \cX^{n\eps}$ which are jointly typical with $\vy(u)$. Define $\slk{TYP} = \frac{1}{n} (|\cX| - 1) \log(n \eps +1)$, which goes to $0$ as $n \to \infty$. 
We have $|\cA_{k}| \le 2^{n \eps H(\bfx | \bfy, \bfu=u) }$. From the code construction, for each $u \in \cU^{\le \alpha}$ we have $|\cT_{n \eps}(P_{\bfx|\bfu = u})| \ge 2^{n \eps ( H(\bfx|\bfu=u) - \slk{TYP} ) }$. 

We want to show that with high probability over the construction of the code that the list size will be guaranteed. To emphasize that codewords are random variables, we write $\rvx_{(m, r)}^{(u)}$ for the codeword symbols corresponding to to message $m$, randomness $r$, and chunk $u$. 
Fix a message $m$ and let $\cC^{(u)_{m}} = \{ \rvx^{(u)}_{(m, r)} : r \in \cR\}$ be the set of codeword chunks representing message $m$ in chunk $u$. The probability over the code construction that a single $\rvx^{(u)}_{(m, r)}$ uniformly selected from $\cT_{n\eps}(P_{\bfx|\bfu = u})$ falls in $\cA_{k}(\vy(u))$ is
	\begin{align}
	\eta_u \leq \frac{ 2^{n \eps H(\bfx | \bfy, \bfu=u)  }}{ 2^{n \eps ( H(\bfx|\bfu=u) -  \slk{TYP} )}} = 2^{- n \eps ( I(\bfx ; \bfy | \bfu = u) - \slk{TYP} )}
	\end{align}
Hence by a union bound, the probability that any element of $\cC^{(u)_{m}}$ falls in $\cA_{u}(\vy(u))$ is at most $N \eta_u$. 

With some abuse of notation, we use $\exp(a) = 2^{a}$. Since the codebooks in each chunk are chosen independently, the probability that there exists an $r$ such that $\rvx^{(u)}_{(m, r)}$ is jointly typical with $\vy(u)$ in each chunk $u \le \alpha K$ is at most
	\begin{align}
	\eta &= \prod_{u=1}^{\alpha K} N 2^{- n \eps ( I(\bfx ; \bfy | \bfu = u ) - 2 \slk{TYP} ) } \\ 
	&= \exp\left( - n \frac{1}{K} \sum_{u=1}^{ \alpha K } I(\bfx ; \bfy | \bfu = u)  + \alpha \slk{TYP} n + \gamma n  \right) \\
	&= \exp\left( - n \left( I(P_{\bfx|\bfu}, V_{\bfs|\bfx, \bfu^{\le\alpha}}) - \slk{TYP} - \gamma \right) \right).
	\end{align}

%\mikel{Should the number of types $\phi$ below depend on $\alpha$ (or $n\eps$)?}
%
%Now, taking a union bound over $\cF_\alpha(P_{\bfx|\bfu})$ and setting $\phi(n) = ( |\cS||\cX| + 1) \frac{\log n}{n}$ we can set
%	\begin{align}
%	\zeta &= \exp\left( - n \left( \min_{\cF_\alpha(P_{\bfx|\bfu})} I(P_{\bfx|\bfu}, V_{\bfs|\bfx, \bfu^{\le\alpha}}) - \gamma - ( |\cS||\cX| + 1) \frac{\log n}{n} \right) \right)\\
%	&= \exp\left( -n \left( R_0 - \gamma - \phi(n) \right) \right).
%	\end{align}
%This is an upper bound on the probability that a message $m$ ends up in the list produced by Bob. 

We next have to find an upper bound on the list size. The probability that a particular set of $L+1$ messages is in the list is upper bounded by $\eta^{L+1}$. Taking a union bound over all possible lists (there are at most $2^{nR (L+1)}$) and all $2^{n \log |\cY|}$ choices for $\vy$, the probability that codebook generation places more than $L+1$ messages in the list produced by Bob is at most
	\begin{align}
	\psi = \exp\left( - n \left( I(P_{\bfx|\bfu}, V_{\bfs|\bfx, \bfu^{\le\alpha}}) (L+1) - R (L+1) - (L+1) \slk{TYP} -  (L+1) \gamma - \log |\cY| \right) \right)
	\end{align}
By assumption, $\slk{UB} \le I(P_{\bfx|\bfu}, V_{\bfs|\bfx, \bfu^{\le\alpha}}) - R \le \slk{LB}$, so
	\begin{align}
	\psi \le \exp\left( - n (L+1) (\slk{UB} - \slk{TYP} -  \gamma) + n \log |\cY| \right)
	\end{align}
To make this probability less than $1$ we can choose
	\begin{align}
	L > \frac{ \log |\cY| }{ \slk{UB} - \slk{TYP} - \gamma  }.
	\end{align}
Now set $\gamma < \slk{UB}/2$. For sufficiently large $n$ we can make $\slk{TYP}$ as small as we like and set
	\begin{align}
	L > \frac{ 2 \log |\cY| }{  \slk{UB} }.
	\end{align}
to make the probability of not generating a good code smaller than $2^{-n \gamma/2}$.
To guarantee that that the above holds for any pair $(\alpha,V_{\bfs|\bfx,\bfu^{\le\alpha}})$ that may be used in our decoding procedure, we apply the union bound over 
$2^{\phi(n)}$ such pairs where $\phi(n) \leq (|\cS||\cX| + 1)\log (n\eps) + \log(1/\eps)$.
All in all, for sufficiently large $n$, we conclude that  the probability of not generating a good code is smaller than $2^{-n \gamma/4}$.
\end{IEEEproof}

%\begin{claim}[List size]
%	\label{claim:listsize}
%	With probability $\pr{1}$ over code design, for any pair $(\alpha,V_{\bfs|\bfx,\bfu^{\le\alpha}})$ the size of  
%	$\cL(P_{\bfx|\bfu},V_{\bfs|\bfx,\bfu^{\le\alpha}})$ is at most $\listsize$. 
%\end{claim}
%
%\mikel{A connection between $\slk{UB}$ and $\listsize$ we be determined in analysis above.}

For a distribution $P$, let $\Delta_\slackrate(P)$ includes all distributions within $\ell_\infty$ distance $\slackrate$ from $P$.

\begin{claim}[Typicality properties]
	\label{claim:listprop}
	With probability $\pr{2}$ over code design it holds for any message $m$ and any list $\cL$ of size at most $\listsize$ obtained in the decoding process, that
	\begin{eqnarray}
	\left|\left\{
	(r_{\alpha K+1},\ldots, r_K) \in \cR^{K(1-\alpha)} :
	\forall m' \in \cL,\  \forall u\in\cU^{>\alpha},\,\ \forall (r'_{\alpha K+1},\ldots, r'_K) \in \cR^{K(1-\alpha)}, \,\  \right .\right .\label{eq:claim7}\\ 
\left .\left .	
%\ |\type_{\vx_{(m,r_u)}^{(u)},\vx'^{(u)}_{(m',r'_u)}}(x,x') - P_{\bfx|\bfu=u}(x)P_{\bfx|\bfu=u}(x')| \leq \slk{x|u}\cdot {\bf 1}_{P_{\bfx|\bfu = u}(x)=0} \cdot {\bf 1}_{P_{\bfx|\bfu = u}(x')=0}
\type_{\vx_{(m,r_u)}^{(u)},\vx'^{(u)}_{(m',r'_u)}} \in \Delta_{\slk{x|u}}(P_{\bfx|\bfu = u}^{\ot2})
	\right\}\right|	
	\geq (1-\pr{3})|\cR|^{K(1-\alpha)} \nonumber 		
	\end{eqnarray}
\end{claim}

\begin{IEEEproof}
%Fix any $u \in \cU^{>\alpha}$. 
For any $m'\in \cL$ with $m'\neq m$ and fixed $(r_{\alpha K+1},\ldots, r_K)$ and $(r'_{\alpha K+1},\ldots, r'_K)$, by Sanov's theorem followed by Pinsker's inequality, and finally by setting $\slk{x|u} = 3\varepsilon$, we have that for any sufficiently large $n$,
\begin{eqnarray}
\lefteqn{\Pr_{\bf \vx_{(m,r_u)}^{(u)}} \left (\exists u \in \cU^{>\alpha},\ \type_{\vx_{(m,r_u)}^{(u)},\vx'^{(u)}_{(m',r'_u)}} \notin \Delta_{\slk{x|u}}(P_{\bfx|\bfu = u}^{\ot2}) \right   )} \label{eq:atyp_j_type}\\
 & \leq & \exp \left(-n\varepsilon D (P_{\bfx,\bfx'|\bfu = u}^{\ot2} \| P_{\bfx|\bfu = u}^{\ot2}) \right) 	+ c_1\log(n)\nonumber \\
& \leq & \exp \left(-n\varepsilon \frac{\|P_{\bfx,\bfx'|\bfu = u}^{\ot2} - P_{\bfx|\bfu = u}^{\ot2}\|_1^2}{2\ln(2)} \right) 	+ c_1\log(n) \nonumber \\
& \leq & \exp \left(-3n\varepsilon^3  \right). \nonumber
\end{eqnarray}
(In the inequalities above, $c_1$ is a constant independent of the blocklength $n$.)

Taking a union bound over all $|\cR|^{K} = 2^{n\gamma/\varepsilon}=2^{n\varepsilon^3}$ values of $(r'_{\alpha K+1},\ldots, r'_K)$ and all $|\cL|$ elements $m'$ in $\cL$, we have that the probability that the condition in~\eqref{eq:claim7} is not satisfied for fixed $(r_{\alpha K+1},\ldots, r_K)$ is at most $|\cL|\exp \left(-2n\varepsilon^3  \right) < \exp \left(-\frac{3}{2}n\varepsilon^3  \right)$ (by bounding $|\cL|$ from above, very loosely, by Claim~\ref{claim:listsize}).

Hence for any given $m$ the expected fraction of $(r_{\alpha K+1},\ldots, r_K)$ not satisfying the event of~\eqref{eq:claim7} is at most $\exp \left(-\frac{3}{2}n\varepsilon^3  \right)$. By standard tail inequalities, and additional ideas appearing in the analysis of Claim III.21 in \cite{chen2019capacity}, the probability that the fraction of $(r_{\alpha K+1},\ldots, r_K)$ not satisfying the event of~\eqref{eq:claim7} exceeds $\exp \left(-n\varepsilon^3  \right)=1-\pr{3}$ is at most $\exp \left(-n^2 \right)$.

The above computation is for a given $\alpha$, $m$, and list $\cL$. Hence to conclude the asserted claim we take the union bound over all possible $\alpha$, $m$, and list $\cL$, to obtain, for sufficiently large $n$, $\pr{2}=1-\exp(-n^2/2)$.
%the number of $(r_{\alpha K+1},\ldots, r_K) \in \cR^{K(1-\alpha)}$ such that some $r_u$ does not satisfy the event of~\eqref{eq:atyp_j_type} is at most $n^{K(1-\alpha)} < n^{1/\varepsilon}$, with probability at least $1-\pr{2}/\varepsilon$. Hence $\pr{3} \leq n^{1/\epsilon}/2^{\varepsilon^3 n} \leq 2^{-\varepsilon^3 n/2}$
\end{IEEEproof}

%\mikel{Slackness parameter $\slk{x|u}$ will be determined by the analysis above. Also, the statement should change slightly to also guarantee that for $x$ for which $P_{\bfx|\bfu = u}(x)=0$ we have for any $x'$ that 
%$\type_{\vx_{(m,r_u)}^{(u)},\vx'^{(u)}}(x,x')=0$. This will be needed later.}

From this point on, we assume our code $\cC$ holds the properties of Claims~\ref{claim:listsize} and \ref{claim:listprop}.

\subsubsection{Decoding analysis}
Consider a  message $m$, corresponding codeword transmitted by Alice $\vx_{(m,r)} $, and a jamming vector $\vs$ used by James to obtain the received $\vy$.
Let $V^*_{\bfs|\bfx,\bfu} \in \cF(P_{\bfx|\bfu})$ satisfy 
$\forall u\in\cU$, $\type_{\vx^{(u)},\vs^{(u)},\vy^{(u)}} =P_{\bfx|\bfu = u} V^*_{\bfs|\bfx,\bfu} W_{\bfy|\bfx,\bfs}$.
Let $\alpha^*$ satisfy $R \in [ I(P_{\bfx|\bfu},V_{\bfs|\bfx,\bfu^{\le\alpha^*}})-\slk{LB},  I(P_{\bfx|\bfu},V_{\bfs|\bfx,\bfu^{\le\alpha^*}})-\slk{UB}]$. 
By the facts that $R = C_{K,\slackrate} - \slackrate \leq \min_{ V_{\bfs|\bfx,\bfu}\in\cF(P_{\bfx|\bfu})} I(P_{\bfx|\bfu}, V_{\bfs|\bfx,\bfu}) - \slackrate \leq I(P_{\bfx|\bfu}, V^*_{\bfs|\bfx,\bfu}) - \slackrate$; and that $\slk{LB}-\slk{UB} \geq \frac{1}{K}\log{|\cX|}$ we conclude that such $\alpha^* \in [K-1]$ exists.

\begin{theorem}[Correct decoding with $(\alpha^*,V^*_{\bfs|\bfx,\bfu^{\le\alpha}})$]
	\label{thm:correct_decoding}
	%With probability $1-\pr{3}$ over the stochasticity of Alice, 
	When Bob decodes using the pair $(\alpha^*,V^*_{\bfs|\bfx,\bfu^{\le\alpha}})$, the message $m$ will pass  the two-step decoding of Bob.
\end{theorem}	

\begin{proof}(of Theorem~\ref{thm:correct_decoding})
Follows directly by the suggested decoding scheme.
\end{proof}

\begin{theorem}[No decoding when $\alpha \le \alpha^*$]
	\label{thm:no_decoding}
	When Bob decodes using a pair $(\alpha,V_{\bfs|\bfx,\bfu^{\le\alpha}})$ in which $\alpha \le \alpha^*$, if for the transmitted message $\vx_{(m,r)} $ Alice's stochasticity $r$ of falls into the set of size $(1-\pr{3})|\cR|^{K(1-\alpha)}$ of Claim~\ref{claim:listprop}, no message different from $m$ will be decoded (specifically, no codeword $\vx' \ne \vx$ will pass the test in the unique-decoding step).
\end{theorem}

\begin{proof}(of Theorem~\ref{thm:no_decoding})
Assume in contradiction that there exists $(\alpha,V_{\bfs|\bfx,\bfu^{\le\alpha}})$ in which $\alpha \le \alpha^*$ and $\vx' \ne \vx \in \cL(P_{\bfx|\bfu},V_{\bfs|\bfx,\bfu^{\le\alpha}})$ such that the message $m'  \ne m$ corresponding to $\vx'$ is decoded by Bob. 
Let $\type_{\vx^{(u)},\vx'^{(u)}}= P_{\bfx,\bfx'|\bfu=u}$, $\forall u\in\cU$.
In particular, there exists a vector $\vs'$ and  $V_{\bfs|\bfx,\bfu^{>\alpha}}$ for which $(V_{\bfs|\bfx,\bfu^{\le\alpha}},V_{\bfs|\bfx,\bfu^{>\alpha}}) \in \cF_\alpha(P_{\bfx|\bfu})$ and $\forall u\in\cU^{>\alpha}$, $\type_{\vx'^{(u)},\vs'^{(u)},\vy^{(u)}} = P_{\bfx|\bfu = u} V_{\bfs|\bfx,\bfu^{>\alpha} = u} W_{\bfy|\bfx,\bfs}$.
I.e., corrupting $\vx'$ using $\vs'$ results in the received word $\vy$.

Using the pair $\vx$, $\vx'$ let $V^*_{\bfs|\bfx,\bfx',\bfu^{>\alpha}}$ satisfy 
$$
\forall (x,x',s, y),\  \forall u\in\cU^{>\alpha},\,\type_{\vx^{(u)},\vx'^{(u)},\vs^{(u)},\vy^{(u)}}(x,x',s,y) = %\Delta_{\slk{x|u}}(
P_{\bfx,\bfx'|\bfu=u}(x,x') V_{\bfs|\bfx,\bfx',\bfu^{>\alpha} = u}^*(s|x,x') W_{\bfy|\bfx,\bfs}(y|x,s)
$$
and  let $V_{\bfs|\bfx,\bfx',\bfu^{>\alpha}}$ satisfy 
$$
\forall (x,x',s, y),\  \forall u\in\cU^{>\alpha},\,\type_{\vx^{(u)},\vx'^{(u)},\vs'^{(u)},\vy^{(u)}}(x,x',s,y) = %\Delta_{\slk{x|u}}
P_{\bfx,\bfx'|\bfu=u}(x',x) V_{\bfs|\bfx,\bfx',\bfu^{>\alpha} = u}(s|x',x) W_{\bfy|\bfx,\bfs}(y|x',s)
$$
It follows that $(V_{\bfs|\bfx,\bfu^{\le\alpha}}, V_{\bfs|\bfx,\bfx',\bfu^{>\alpha}})\in\cF_\alpha(P_{\bfx|\bfu})$ and  $(V^*_{\bfs|\bfx,\bfu^{\le\alpha}}, V^*_{\bfs|\bfx,\bfx',\bfu^{>\alpha}})\in\cF_\alpha(P_{\bfx|\bfu})$.

%Let 
%\begin{align}
%\label{ex:vbar}
%\bar{V}_{\bfs|\bfx,\bfu^{\le\alpha}}= \argmin_{
%	\substack{V_{\bfs|\bfx,\bfu^{\le\alpha}} \\ 
%		R  \geq  I(P_{\bfx|\bfu},V_{\bfs|\bfx,\bfu^{\le\alpha}})-\slk{LB} }}
%{\cost([P_{\bfx|\bfu}V_{\bfs|\bfx,\bfu^{\le\alpha}}]_{\bfs^{\le\alpha}})}
%\end{align}

\begin{claim}
	\label{claim:convex_opt}
	There exists 
\begin{align}
\label{ex:vbar}
\bar{V}_{\bfs|\bfx,\bfu^{\le\alpha}} \in \argmin_{
	\substack{V_{\bfs|\bfx,\bfu^{\le\alpha}} \\ 
		R  \geq  I(P_{\bfx|\bfu},V_{\bfs|\bfx,\bfu^{\le\alpha}})-\slk{LB} }}
{\cost([P_{\bfx|\bfu}V_{\bfs|\bfx,\bfu^{\le\alpha}}]_{\bfs^{\le\alpha}})}
\end{align}
such that
\begin{enumerate}
	\item $(\bar{V}_{\bfs|\bfx,\bfu^{\le\alpha}}, V_{\bfs|\bfx,\bfx',\bfu^{>\alpha}})\in\cF_\alpha(P_{\bfx|\bfu})$.
	\item $(\bar{V}_{\bfs|\bfx,\bfu^{\le\alpha}}, V^*_{\bfs|\bfx,\bfx',\bfu^{>\alpha}})\in\cF_\alpha(P_{\bfx|\bfu})$.
	\item $R \in [ I(P_{\bfx|\bfu},\bar{V}_{\bfs|\bfx,\bfu^{\le\alpha}})-\slk{LB}, I(P_{\bfx|\bfu},\bar{V}_{\bfs|\bfx,\bfu^{\le\alpha}})-\slk{UB}]$.
\end{enumerate}
\end{claim}

\begin{proof}(of Claim~\ref{claim:convex_opt})
Let $\bar{V}_{\bfs|\bfx,\bfu^{\le\alpha}}$ be any distribution satisfying  (\ref{ex:vbar}).
The first claim that $(\bar{V}_{\bfs|\bfx,\bfu^{\le\alpha}}, V_{\bfs|\bfx,\bfx',\bfu^{>\alpha}})\in\cF_\alpha(P_{\bfx|\bfu})$ follows from the fact that $ R \in  [I(P_{\bfx|\bfu},V_{\bfs|\bfx,\bfu^{\le\alpha}}) - \slk{LB}, I(P_{\bfx|\bfu},V_{\bfs|\bfx,\bfu^{\le\alpha}})-\slk{UB}] $ and thus $\cost([P_{\bfx|\bfu}\bar{V}_{\bfs|\bfx,\bfu^{\le\alpha}}]_{\bfs^{\le\alpha}}) \leq 
\cost([P_{\bfx|\bfu}V_{\bfs|\bfx,\bfu^{\le\alpha}}]_{\bfs^{\le\alpha}})$ (here we use the nature of $\lambda_\bfs$).
% and the fact that $\slk{OPT} \geq \slk{LB}$ \mikel{Perhaps one can set $\slk{OPT} = \slk{LB}$}).
A similar analysis holds  for the second claim, $(\bar{V}_{\bfs|\bfx,\bfu^{\le\alpha}}, V^*_{\bfs|\bfx,\bfx',\bfu^{>\alpha}})\in\cF_\alpha(P_{\bfx|\bfu})$, with the additional observation that as $\alpha \le \alpha^*$ we have  
 $I(P_{\bfx|\bfu},V^*_{\bfs|\bfx,\bfu^{\le\alpha}}) \leq I(P_{\bfx|\bfu},V^*_{\bfs|\bfx,\bfu^{\le\alpha^*}})\in[ R+\slk{UB},R+\slk{LB}]$.
 
 Finally, the third claim follows from standard convexity arguments coupled by our assumption that there exists a zero-cost state $s_0\in \cS$ and a one-to-one mapping $\phi:\cY \rightarrow \cX$ for which for every $x$, $W_{\bfy|\bfx,\bfs}(\phi(x)|x,s_0)=1$.
%Such a state $s_0$ has no cost to James and, accordingly, does not corrupt Bob's view of $x$ (as $x$ can be decoded from $\phi(x)$).
%If such a state $s_0$ does not exist for the given channel $W$, one can replace $W$ by $W'$, replace $\cY$ by $\cY'=\cY \cup \{y_x | x \in \cX\}$, and $\cS$ by $\cS'=\cS \cup \{s_0\}$, where $\{y_x | x \in \cX\} \cap \cY = \phi$, for all $x \in \cX$, $y \in \cY$, $s \in \cS$, $W'(y|x,s)=W(y|x,s)$, and for all $x \in \cX$, $W'(y_x|x,s_0)=1$ (the remaining values of $W'(y|x,s)$ are set to 0).
%As James is stronger when one considers the AVC defined by $\cX'$, $\cY'$, and $W'$ when compared to that defined by $\cX$, $\cY$, and $W$, the achievability proof at hand may assume the former without loss of generality.
%\mikel{Need to rethink the argument again - we may actually need to prove that the change does not reduce rate as we cannot use $W'$ in an isolated manner in this part of the proof and not in others ... which will eventually lead to capacity expressions that depend of $W'$.} 
More formally, consider  ${V}^{(0)}_{\bfs|\bfx,\bfu^{\le\alpha}}$ for which for every $x$ and $u \leq \alpha$, ${V}^{(0)}_{\bfs|\bfx,\bfu^{\le\alpha}}(s_0|x,u)=1$. 
 It holds that  $\cost([P_{\bfx|\bfu}V^{(0)}_{\bfs|\bfx,\bfu^{\le\alpha}}]_{\bfs^{\le\alpha}})=0$, and 
 that 
$I(P_{\bfx|\bfu},V^{(0)}_{\bfs|\bfx,\bfu^{\le\alpha}}) \geq I(P_{\bfx|\bfu},V_{\bfs|\bfx,\bfu^{\le\alpha}}) \geq R+\slk{UB}$.
% 
% as 
% $R \in \cR_\slackrate$ we have that
% $R \le \min_{ V_{\bfs|\bfx,\bfu}\in\cF(P_{\bfx|\bfu})} I(P_{\bfx|\bfu}, V_{\bfs|\bfx,\bfu}) - \slk{GOOD} \leq I(P_{\bfx|\bfu}, V^{(0)}_{\bfs|\bfx,\bfu}) - \slk{GOOD}$.
%As  $\slk{GOOD} \geq \slk{OPT}$, the latter in turn implies that   
% $R \leq I(P_{\bfx|\bfu}, V^{(0)}_{\bfs|\bfx,\bfu}) - \slk{OPT}$. 
Assume in contradiction that 
for every $\bar{V}_{\bfs|\bfx,\bfu^{\le\alpha}}$ satisfying  (\ref{ex:vbar}) it holds that 
$R \not\in [ I(P_{\bfx|\bfu},\bar{V}_{\bfs|\bfx,\bfu^{\le\alpha}})-\slk{LB}, I(P_{\bfx|\bfu},\bar{V}_{\bfs|\bfx,\bfu^{\le\alpha}})-\slk{UB}]$ or equivalently that
$R >  I(P_{\bfx|\bfu},\bar{V}_{\bfs|\bfx,\bfu^{\le\alpha}})-\slk{UB}$.
Consider one such $\bar{V}_{\bfs|\bfx,\bfu^{\le\alpha}}$ and 
define the convex combination 
 $V^\lambda_{\bfs|\bfx,\bfu^{\le\alpha}}=\lambda\bar{V}_{\bfs|\bfx,\bfu^{\le\alpha}}+(1-\lambda)V^{(0)}_{\bfs|\bfx,\bfu^{\le\alpha}}$
 for an arbitrarily small $\lambda >0$.
From the linear nature of our cost function, it follows that  
$\cost([P_{\bfx|\bfu}V^\lambda_{\bfs|\bfx,\bfu^{\le\alpha}}]_{\bfs^{\le\alpha}}) \leq \cost([P_{\bfx|\bfu}\bar{V}_{\bfs|\bfx,\bfu^{\le\alpha}}]_{\bfs^{\le\alpha}})$ 
with equality only if $\cost([P_{\bfx|\bfu}\bar{V}_{\bfs|\bfx,\bfu^{\le\alpha}}]_{\bfs^{\le\alpha}})=0$.

For $\cost([P_{\bfx|\bfu}\bar{V}_{\bfs|\bfx,\bfu^{\le\alpha}}]_{\bfs^{\le\alpha}})>0$, we have that 
$\cost([P_{\bfx|\bfu}V^\lambda_{\bfs|\bfx,\bfu^{\le\alpha}}]_{\bfs^{\le\alpha}}) < \cost([P_{\bfx|\bfu}\bar{V}_{\bfs|\bfx,\bfu^{\le\alpha}}]_{\bfs^{\le\alpha}})$.
Moreover, from the continuity of $I(P_{\bfx|\bfu},V^\lambda_{\bfs|\bfx,\bfu^{\le\alpha}})$ (in $\lambda$), for sufficiently small $\lambda>0$ we have that $R >  I(P_{\bfx|\bfu},\bar{V}_{\bfs|\bfx,\bfu^{\le\alpha}})-\slk{UB} \geq I(P_{\bfx|\bfu},\bar{V}_{\bfs|\bfx,\bfu^{\le\alpha}})-\slk{LB}$ in contradiction to the definition of  $\bar{V}_{\bfs|\bfx,\bfu^{\le\alpha}}$ in (\ref{ex:vbar}).
% $\cost([P_{\bfx|\bfu}V_{\bfs|\bfx,\bfu^{\le\alpha}}]_{\bfs^{\le\alpha}})=0$, and  $I(P_{\bfx|\bfu},V^*_{\bfs|\bfx,\bfu^{\le\alpha}}) \leq I(P_{\bfx|\bfu},V^{(0)}_{\bfs|\bfx,\bfu^{\le\alpha}}) \leq equals $s_0$.
If $\cost([P_{\bfx|\bfu}\bar{V}_{\bfs|\bfx,\bfu^{\le\alpha}}]_{\bfs^{\le\alpha}})=0$, again from the continuity of $I(P_{\bfx|\bfu},V^\lambda_{\bfs|\bfx,\bfu^{\le\alpha}})$ (in $\lambda$), there exists a $\lambda \in (0,1)$ for which 
$R = I(P_{\bfx|\bfu},V^\lambda_{\bfs|\bfx,\bfu^{\le\alpha}})-\slk{UB} \in [ I(P_{\bfx|\bfu},V^\lambda_{\bfs|\bfx,\bfu^{\le\alpha}})-\slk{LB}, I(P_{\bfx|\bfu},V^\lambda_{\bfs|\bfx,\bfu^{\le\alpha}})-\slk{UB}]$, inplying that $V^\lambda_{\bfs|\bfx,\bfu^{\le\alpha}}$ satisfies (\ref{ex:vbar}) and can be used to conclude the assertion.
 \end{proof}

Using Claim~\ref{claim:convex_opt} and recalling that $R \in \cR_\slackrate$ we conclude that there exists $ u\in\cU^{>\alpha} $ such that $ (V_{\bfs|\bfx,\bfx',\bfu^{>\alpha} = u}, V_{\bfs|\bfx,\bfx',\bfu^{>\alpha} = u}^*)\not\in\cV'_\slackrate$.
However, this latter fact is in contradiction with the definition of $V^*_{\bfs|\bfx,\bfx',\bfu^{>\alpha}}$ and $V_{\bfs|\bfx,\bfx',\bfu^{>\alpha}}$.
In particular, we have for $\bfu=u$ that:
$$
\forall (x,x', y),\ \type_{\vx^{(u)},\vx'^{(u)},\vy^{(u)}}(x,x',y) = %\Delta_{\slk{x|u}}(
\sum_{s \in \cS}P_{\bfx,\bfx'|\bfu=u}(x,x') V_{\bfs|\bfx,\bfx',\bfu^{>\alpha} = u}^*(s|x,x') W_{\bfy|\bfx,\bfs}(y|x,s),
$$
$$
\forall (x,x',s, y),\ \type_{\vx^{(u)},\vx'^{(u)},\vy^{(u)}}(x,x',y) = %\Delta_{\slk{x|u}}
\sum_{s \in \cS}P_{\bfx,\bfx'|\bfu=u}(x',x) V_{\bfs|\bfx,\bfx',\bfu^{>\alpha} = u}(s|x',x) W_{\bfy|\bfx,\bfs}(y|x',s).
$$
%$$
%\forall (x,x',s, y),\,\type_{\vx^{(u)},\vx'^{(u)},\vs^{(u)},\vy^{(u)}}(x,x',s,y) = P_{\bfx|\bfu = u}(x)P_{\bfx|\bfu = u}(x') V_{\bfs|\bfx,\bfx',\bfu^{>\alpha} = u}(s|x,x') W_{\bfy|\bfx,\bfs}(y|x,s),
%$$ 
%$$
%\forall (x,x',s', y),\,\type_{\vx^{(u)},\vx'^{(u)},\vs^{(u)},\vy^{(u)}}(x,x',s',y) = P_{\bfx|\bfu = u}(x)P_{\bfx|\bfu = u}(x') V_{\bfs|\bfx,\bfx',\bfu^{>\alpha} = u}(s'|x',x) W_{\bfy|\bfx,\bfs}(y|x',s'),
%$$
In addition, by the assumption that Claim~\ref{claim:listprop} holds, we have that
$$
\type_{\vx_{(m,r_u)}^{(u)},\vx'^{(u)}} = P_{\bfx,\bfx'|\bfu=u} \in \Delta_{\slk{x|u}} (P_{\bfx|\bfu = u}^{\ot2}).
$$
Moreover, from our code design, we have for $x$ for which $P_{\bfx|\bfu = u}(x)=0$ we have for any $x'$ that $P_{\bfx,\bfx'|\bfu=u}(x,x')=0$. 
This implies that 
$$
\forall (x,x'),\ |P_{\bfx,\bfx'|\bfu=u}(x,x') - P_{\bfx|\bfu=u}(x)P_{\bfx|\bfu=u}(x')| \leq \slk{x|u}\cdot {\bf 1}_{P_{\bfx|\bfu = u}(x) \neq 0} \cdot {\bf 1}_{P_{\bfx|\bfu = u}(x') \neq 0}
$$
Now, the above implies that
%$$
%\forall (x,x',s,y),\  
%P_{\bfx,\bfx'|\bfu=u}(x,x') V_{\bfs|\bfx,\bfx',\bfu^{>\alpha} = u}^*(s|x,x') W_{\bfy|\bfx,\bfs}(y|x,s)=
%P_{\bfx,\bfx'|\bfu=u}(x',x) V_{\bfs|\bfx,\bfx',\bfu^{>\alpha} = u}(s|x',x) W_{\bfy|\bfx,\bfs}(y|x',s)
%$$
%which, in turn, implies that 
$$
\forall (x,x',y),\ 
P_{\bfx,\bfx'|\bfu=u}(x,x')\sum_{s\in\cS}V_{\bfs|\bfx,\bfx',\bfu^{>\alpha} = u}^*(s|x,x')W_{\bfy|\bfx,\bfs}(y|x,s) 
=
P_{\bfx,\bfx'|\bfu=u}(x',x)\sum_{s\in\cS}V_{\bfs|\bfx,\bfx',\bfu^{>\alpha} = u}(s'|x',x)W_{\bfy|\bfx,\bfs}(y|x',s) 
$$
and thus $\forall (x,x',y)$ for which $P_{\bfx|\bfu = u}(x) \ne 0$ and $P_{\bfx|\bfu = u}(x') \ne 0$:
\begin{align*}
\left|\sum_{s\in\cS}V_{\bfs|\bfx,\bfx',\bfu^{>\alpha} = u}^*(s|x,x')W_{\bfy|\bfx,\bfs}(y|x,s) -\sum_{s\in\cS}V_{\bfs|\bfx,\bfx',\bfu^{>\alpha} = u}(s'|x',x)W_{\bfy|\bfx,\bfs}(y|x',s) \right| \leq
\frac{2\slk{x|u}}{P_{\bfx|\bfu=u}(x)P_{\bfx|\bfu=u}(x')}
\end{align*}
In contradiction to 
$ (V_{\bfs|\bfx,\bfx',\bfu^{>\alpha} = u}, V_{\bfs|\bfx,\bfx',\bfu^{>\alpha} = u}^*)\not\in\cV'_\slackrate$ as
$\slackrate =   \frac{2\slk{x|u}}{\slk{MIN}^2} \geq \frac{2\slk{x|u}}{\min_{x}\{P_{\bfx|\bfu=u}^2(x)\}}$ where the denominator-minimization is taken over $x$ for which $P_{\bfx|\bfu = u}(x) \ne 0$.
This concludes our proof.

%\mikel{ML: To connect to the main theorem, we may need to optimize only over a quantized version of $P_{\bfx|\bfu = u}(x)$ in which any non-zero probabilities are at least some fixed parameter $\slk{Non-Zero}$. The loss in rate is then bounded by a function of  $\slk{Non-Zero}$ which should be specified.}
\end{proof}

% !TEX root = causal_general.tex

\section{Converse}
\label{sec:converse}

To prove Theorem~\ref{thm:converse_main} it suffices to prove the following theorem.
\begin{theorem}
\label{thm:converse}
Let $ (\cX,\cS,\cY,\ipcstr,\stcstr,W_{\bfy|\bfx,\bfs}) $ be a causal channel satisfying \Cref{itm:assumption-alpha,itm:assumption-ip-constr,itm:assumption-st-constr,itm:assumption-ch} in \Cref{sec:global-assumptions}. 
For any $ \delta>0, K\ge \frac{2\log|\cX|}{\delta} + \binom{|\cX|+1}{2} $ and any code $ \cC\subset\cX^n $ satisfying
\begin{enumerate}
	\item $ \cC $ is equipped with an encoder-decoder pair $ (\enc,\dec)\in\Delta(\cX^n|\cM)\times\Delta(\cM|\cY^n) $ (both potentially stochastic);

	\item $ \type_{\vx}\in\ipcstr $ for every $ \vx\in\cC $;

	\item $ R(\cC)\ge \ol{C}_{K,\delta} + \delta $, 
\end{enumerate}
the average error probability is at least
\begin{align}
\sup_{\jam} P_{\e,\avg}(\enc,\dec,\jam) &\ge f(\delta) \notag 
\end{align}
for some $ f(\delta)>0 $ such that $ f(\delta)\xrightarrow{\delta\to0}0 $. 
The supremum above is taken over all feasible causal jamming strategies as per \Cref{def:jamming}. 
The rate bound $ \ol{C}_{K,\delta} $ is defined (for general $P_\bfu$) as follows: 
\begin{align}
	\ol{C}_{K,\delta} \coloneqq \max_{\substack{(P_\bfu, P_{\bfx|\bfu})\in\Delta(\cU)\times\Delta(\cX|\cU) \\ \sqrbrkt{P_\bfu P_{\bfx|\bfu}}_\bfx\in\ipcstr}} 
	\min\Bigg\{
		& \min_{ V_{\bfs|\bfx,\bfu}\in\cF_{\slk{state}}(P_\bfu, P_{\bfx|\bfu})} I(P_{\bfx|\bfu}, V_{\bfs|\bfx,\bfu}) , \notag \\
		& \min_{\substack{(\alpha, (V_{\bfs|\bfx,\bfu^{\le\alpha}}, V_{\bfs|\bfx,\bfx',\bfu^{>\alpha}}))\in\cA\times\cF_{\alpha,\slk{state}}(P_\bfu,P_{\bfx|\bfu}) \\ 
		\forall u\in\cU^{>\alpha},\,  V_{\bfs|\bfx,\bfx',\bfu^{>\alpha} = u}\in\cV}}
		I(P_{\bfu^{\le\alpha}}, P_{\bfx|\bfu^{\le\alpha}}, V_{\bfs|\bfx,\bfu^{\le\alpha}}) 
	\Bigg\} , \notag 
\end{align}
where $ K_1\ge\frac{2\log|\cX|}{\delta} $, $ K_2\ge\binom{|\cX|+1}{2} $, $ \cU^{\le\alpha} = \{u_1,\cdots,u_{K_1}\} $, $ \cU^{>\alpha} = \{v_1,\cdots,v_{K_2}\} $, $ P_{\bfu^{\le\alpha}} = \unif(\cU^{\le\alpha}) $, $ P_{\bfu^{>\alpha}}\in\Delta(\cU^{>\alpha}) $, $ \cU = \cU^{\le\alpha}\sqcup\cU^{>\alpha} $, $ K = K_1+K_2 $, $ P_{\bfu}\in\Delta(\cU) $ defined as
\begin{align}
P_{\bfu}(u) &= \begin{cases}
\alpha P_{\bfu^{\le\alpha}}, & u\in\cU^{\le\alpha} \\
(1-\alpha) P_{\bfu^{>\alpha}(u)}, & u\in\cU^{>\alpha}
\end{cases}, \notag 
\end{align}
and 
$\slk{state} = (K_1|\cX||\cS| + K_2|\cX|^2|\cS|)\delta $. 
\end{theorem}

\begin{remark}
In the above upper bound, the distribution $ P_{\bfu}\in\Delta(\cU) $ may not be uniform. 
Specifically, though its restriction $ P_{\bfu^{\le\alpha}} $ is uniform over $ \cU^{\le\alpha} $, the counterpart $ P_{\bfu^{>\alpha}} $ in the suffix is not necessarily uniform over $ \cU^{>\alpha} $. 
This is because the latter is given by a decomposition of a CP distribution extracted from the code. 
However, without loss of generality, one can always fine quantize the alphabet $ \cU^{>\alpha} $ to make $ P_{\bfu^{>\alpha}} $ uniform at the cost of increasing $ |\cU^{>\alpha}| $. 
The resulting bound will match our achievability bound (\Cref{eqn:capacity}) upon taking $ \delta\to0 $ and we therefore have a capacity characterization. 
\end{remark}

\section{Proof of converse (Theorem~\ref{thm:converse})}
\label{sec:proof-converse}

Let $ \cC\subset\cX^n $ be an arbitrary codebook such that $ \type_{\vx}\in\ipcstr $ for every $ \vx\in\cC $. 
Suppose $ R(\cC) = R \ge \ol{C}_{K,\delta} +\delta$. 
Let $ M\coloneqq2^{nR} $ and $ \cM\coloneqq[M] $. 
Without loss of generality, the (potentially stochastic) encoder associated with $ \cC $ can be assumed to take the form
$\enc\colon\cM\times\cR\to\cX^n$
for a certain finite alphabet $ \cR $. 
The stochastic encoding of a message $ m\in\cM $ is then given by $ \enc(m,\bfr)\in\cX^n $ where $ \bfr\sim\unif(\cR) $. 
See \cite{dey-causal-it2013} for why the randomness in $ \enc $ can be assumed to come from $ \bfr\sim\unif(\cR) $. 
Let $ N\coloneqq|\cR| $. 
For notational convenience, we will also use $ \vx_{(m,r)} $ to denote $ \enc(m,r) $. 
We use $ P_{\vbfx|\bfm}\in\Delta(\cX^n|\cM) $ to denote the distribution over codewords induced by $ \enc $:
\begin{align}
P_{\vbfx|\bfm}(\vx|m) &= \frac{1}{N} \card{\curbrkt{ r\in\cR: \enc(m,r) = \vx }}. \notag 
\end{align}
Assume also that the decoder associated with $ \cC $ is $ \dec\in\Delta(\cM|\cY^n) $.

% \begin{lemma}[Chunk-wise approximate constant composition reduction]
% \label{lem:chunkwise-apx-cc}
% \end{lemma}

Let $ K \ge \frac{2\log|\cX|}{\delta} $ and $ \eps \coloneqq 1/K $. 
Set 
\begin{align}
\vu &\coloneqq (\underbracket{1,\cdots,1}_{n\eps},\underbracket{2,\cdots,2}_{n\eps}, \cdots, \underbracket{K,\cdots,K}_{n\eps}) \in [K]^n. \notag 
\end{align}
Let $ \slk{input} = \delta/2 $. 
James examines all chunk-wise approximate constant composition subcodes with a constant fraction of messages. 
To find such subcodes, for each $ \wh{P}_{\bfx|\bfu}\in\Delta(\cX|[K]) $, James computes 
\begin{align}
\cM' &\coloneqq \curbrkt{m\in\cM: \exists\vx\in\cX^n,  P_{\vbfx|\bfm}(\vx|m) \indicator{\norminf{\type_{\vu,\vx} - \unif([K])\times \wh P_{\bfx|\bfu}} \le\slk{input} }>0 }. \label{eqn:msg-set-subcode}
\end{align}
Let $ M' = |\cM'| $. 
If 
\begin{align}
\sum_{m\in\cM'} \sum_{\vx\in\cX^n} P_{\vbfx|\bfm}(\vx|m) \indicator{\norminf{\type_{\vu,\vx} - \unif([K])\times \wh P_{\bfx|\bfu}} \le\slk{input} } &\ge cM 
\label{eqn:cc-reduction-guarantee}
\end{align}
for some constant $ c>0 $ that depends only on $ \slk{input},K $ and $ |\cX| $ (but not on $n$), then the following subcode
\begin{align}
\cC' &\coloneqq \curbrkt{\vx\in\cC: \norminf{\type_{\vu,\vx} - \unif([K])\times \wh P_{\bfx|\bfu}} \le\slk{input} } \label{eqn:subcode} 
\end{align}
will be considered by James. 
Define also the set $\cR'$ of encoder randomness associated with $\cC'$ as:
\begin{align}
\cR' \coloneqq \curbrkt{ r'\in\cR: \exists m'\in\cM',\,\enc(m',r')\in\cC' } . \label{eqn:rand-set-subcode} 
\end{align}

Note that, for each subset $ \cM'\subset\cM $, the associated distribution $ \wh{P}_{\bfx|\bfu} $ may not be unique and therefore the corresponding subcode $ \cC' $ may not be unique. 
However, by a standard Markov argument, a pair $ (\cM',\wh{P}_{\bfx|\bfu}) $ such that \Cref{eqn:cc-reduction-guarantee} holds for some $c>0$ must exist, and therefore there exists at least one subcode $ \cC' $ which James is looking for.

\subsection{Babble-only attack}
\label{sec:converse-babble-only}
Let $ \cU\coloneqq[K] $ and $ P_\bfu \coloneqq \unif(\cU) $. 
James picks a jamming distribution $ V_{\bfs|\bfx,\bfu}\in\cF_{\slk{state}}(P_{\bfx|\bfu}) $. 
James causally observes $ \vx(1),\vx(2),\cdots,\vx(n) $. 
He samples the jamming symbols in the following way: independently for each $ 1\le i\le n $, 
\begin{align}
\vbfs(i) &\sim V_{\bfs|\bfx = \vx(i), \bfu = \vu(i)}. \notag 
\end{align}

\subsection{Babble-and-push attack}
\label{sec:converse-twostage-attack}

\subsubsection{Division point and jamming distributions}

Suppose James is examining a subcode $ \cC'\subset\cC $ with composition $ \wh{P}_{\bfx|\bfu}\in\Delta(\cX|[K]) $ and a message set $ \cM'\subset\cM $ of size $ M' $ such that \Cref{eqn:cc-reduction-guarantee} holds for some constant $c>0$. 

James picks $ \alpha\in\{0,1/K,2/K,\cdots,1\} $. 
Let $ \cU^{\le\alpha}\coloneqq[\alpha K]$.
Let $ \vu^{\le\alpha}\in(\cU^{\le\alpha})^{\alpha n} $ 
be the vectors containing the first $ \alpha n $ entries 
of $ \vu\in[K]^n $.
Let $ P_{\bfu^{\le\alpha}} \coloneqq \unif(\cU^{\le\alpha}) $.
Let $ P_{\bfx|\bfu^{\le\alpha}}\in\Delta(\cX|\cU^{\le\alpha}) $
be the natural restrictions of $ \wh P_{\bfx|\bfu}\in\Delta(\cX|[K]) $, i.e., 
\begin{align}
P_{\bfx|\bfu^{\le\alpha}}(x|u) &\coloneqq \wh P_{\bfx|\bfu}(x|u), \quad \forall u\in\cU^{\le\alpha} . \notag 
\end{align}

Since the subcode $ \cC' $ is chunk-wise constant composition, James can compute the average type $ P_{\bfx^{>\alpha}}\in\Delta(\cX) $ of the suffices: 
\begin{align}
P_{\bfx^{>\alpha}} &\coloneqq \frac{1}{(1-\alpha)K} \sum_{u\in\cU^{>\alpha}} \wh P_{\bfx|\bfu = u}. \notag 
\end{align}
Now James finds a distribution $ \wh{P}_{\bfx,\bfx}\in\cp(P_{\bfx^{>\alpha}}) $ such that the number (or more precisely the probability mass induced by the stochastic encoder) of pairs of codewords from $ \cC' $ whose types are approximately $ \wh{P}_{\bfx,\bfx'} $ is maximized, i.e., 
\begin{align}
\wh{P}_{\bfx,\bfx'} &\coloneqq \argmax_{Q_{\bfx,\bfx'}\in\cp(P_{\bfx^{>\alpha}})} 
\sum_{(m,m')\in\binom{\cM'}{2}} \sum_{(\vx,\vx')\in(\cC')^2} P_{\vbfx|\bfm}(\vx|m) P_{\vbfx|\bfm}(\vx'|m') \indicator{\norminf{\type_{\vx,\vx'} - Q_{\bfx,\bfx'}}\le\slk{CP}} , \notag 
\end{align}
where $ \slk{CP} = \delta/2 $. 
Again by a standard Markov-type argument, note that
\begin{align}
\sum_{(m,m')\in\binom{\cM'}{2}} \sum_{(\vx,\vx')\in(\cC')^2} P_{\vbfx|\bfm}(\vx|m) P_{\vbfx|\bfm}(\vx'|m') \indicator{\norminf{\type_{\vx,\vx'} - \wh{P}_{\bfx,\bfx'}}\le\slk{CP}} &\ge c'\binom{M'}{2} 
\label{eqn:cp-reduction-guarantee}
\end{align}
for some constant $ c'>0 $ independent of $n$. 
Suppose that $ \wh{P}_{\bfx,\bfx'} $ admits a CP decomposition of the following form:
\begin{align}
\wh{P}_{\bfx,\bfx'} &= \sum_{i = 1}^{K'} P_{\bfv}(v_i) P_{\bfx|\bfv = v_i}^{\ot2}, \notag 
\end{align}
for a constant $ K'\subset\bZ_{\ge1} $, a time-sharing distribution $ P_\bfv\in\Delta(\{v_1,\cdots,v_{K'}\}) $ and a conditional distribution $ P_{\bfx|\bfv}\in\Delta(\cX|\{v_1,\cdots,v_{K'}\}) $. 

We then define the overall time-sharing structure. 
Define $ \cU^{>\alpha}\coloneqq\{v_1,\cdots,v_{K'}\} $ and $ \cU \coloneqq \cU^{\le\alpha} \sqcup \cU^{>\alpha} $. 
Let $ P_{\bfu}\in\Delta(\cU) $ be defined as
\begin{align}
P_{\bfu}(u) &\coloneqq \begin{cases}
1/K, & u \in \cU^{\le\alpha} \\
(1-\alpha)P_{\bfv}(u), & u \in \cU^{>\alpha}
\end{cases}. \notag 
\end{align}
Note that $ P_{\bfu} $ is a valid distribution since
\begin{align}
\sum_{ u\in \cU}P_{ \bfu}( u) &= \alpha K\cdot \frac{1}{K} + (1-\alpha) \sum_{v\in\cU^{>\alpha}} P_{\bfv}(v)
= \alpha + (1-\alpha) = 1. \notag 
\end{align}
The input distribution $ P_{\bfx|\bfu}\in\Delta(\cX|\cU) $ under this time-sharing structure can be defined as follows:
\begin{align}
P_{\bfx|\bfu}(x|{u}) &\coloneqq \begin{cases}
{P}_{\bfx|\bfu^{\le\alpha}}(x|{u}), & {u} \in\cU^{\le\alpha} \\
P_{\bfx|\bfv}(x|{u}), & {u} \in\cU^{>\alpha}
\end{cases}. \notag 
\end{align}

For each subcode, James is able to perform the above computation. 
James examines every chunk-wise (approximate) constant composition subcode that satisfies \Cref{eqn:cc-reduction-guarantee} for some constant $ c>0 $ (independent of $n$), and finds $ \alpha\in\cA $ and $ (V_{\bfs|\bfx,\bfu^{\le\alpha}}, V_{\bfs|\bfx,\bfx',\bfv})\in\cF_{\alpha,\slk{state}}(P_\bfu, P_{\bfx|\bfu}) $ that minimizes $ I(P_{\bfu^{\le\alpha}}, P_{\bfx|\bfu^{\le\alpha}}, V_{\bfs|\bfx,\bfu^{\le\alpha}}) $. 
Finally, James picks the subcode $ \cC'\subset\cC $ (as per \Cref{eqn:subcode}) whose induced minimum mutual information is the minimum among all subcodes under consideration. 
Suppose that this $ \cC' $ is associated with a message set $ \cM'\subset\cM $ (as per \Cref{eqn:msg-set-subcode}) and encoding randomness $ \cR'\subset\cR $ (as per \Cref{eqn:rand-set-subcode}). 
By the assumption 
\begin{align}
R\ge \ol{C}_{K,\delta} + \delta , \label{eqn:converse-rate-assumption}
\end{align}
the coding rate $R$ must be larger than the minimum mutual information induced by $\cC'$ by at least an additive factor $ \approx\delta $. 
More precisely, James can guarantee that
\begin{align}
R\ge I(P_{\bfu^{\le\alpha}}, P_{\bfx|\bfu^{\le\alpha}}, V_{\bfs|\bfx,\bfu^{\le\alpha}}) + \delta/2 , 
\label{eqn:rate-guarantee}
\end{align}
where the RHS is the mutual information computed from the subcode $ \cC' $, by setting $ K $ to be sufficiently large, e.g., $ K\ge(2\log|\cX|)/\delta $. 
Indeed, this follows from the following simple bound: for any $ u\in[K] $, 
\begin{align}
\frac{1}{K} I(\bfx_u;\bfy_u) &\le \frac{1}{K} H(\bfx_u) \le \frac{1}{K} \log|\cX| \le \frac{\delta}{2} . \notag 
\end{align}

\begin{remark}
In the suffix, we allow James to ignore the chunk-wise (approximate) constant composition structure of the subcode $ \cC' $. 
In principle, he could have tailored his jamming distribution to $ P_{\bfx|\bfu^{>\alpha}} $ (which is the natural restriction of $ \wh{P}_{\bfx|\bfu} $ to $ u\in[K]\setminus\cU^{\le\alpha} $). 
This will create a two level time-sharing structure in the suffix: the top level being the chunk-wise structure given by $ \bfu^{>\alpha} $, and the bottom level being the time-sharing structure given by CP distributions in each suffix chunk. 
However, this will not have an effective impact on the converse bound we are aiming for. 
This is because $(i)$ the suffix is symmetrized by James and hence does not effectively carry positive amount of information, and $(ii)$ the chunk-wise composition structure does not change the zero-rate threshold (i.e., the Plotkin point) of the suffix. 
Therefore, for simplicity of presentation, we let James forget the chunk-wise composition of $ \cC' $ in the suffix and only use the average composition $ P_{\bfx^{>\alpha}} $. 
\end{remark}

Once James has chosen $ \cC' $ (and the associated $ \cM' $ and $ \cR' $), we will also reveal this particular subcode to Bob. 
This will make Bob have access to knowledge that he is not supposed to have in the original model. 
Under he error probability under the optimal decoder can cannot increase. 
Therefore, our lower bound on error probability will continue to hold for Bob without such knowledge.

\subsubsection{DMC attack in the prefix}
In the prefix, James causally observes $ \vx(1),\vx(2),\cdots,\vx(\alpha n) $. 
He samples the jamming symbols in the following way: independently for each $ 1\le i\le\alpha n $, 
\begin{align}
\vbfs(i) &\sim V_{\bfs|\bfx = \vx(i), \bfu^{\le\alpha} = \vu^{\le\alpha}(i)}. \notag 
\end{align}

\subsubsection{Posterior distribution}
Observing $ \vx^{\le\alpha}\in\cX^{\alpha n} $, James can compute Bob's observation $ \vy^{\le\alpha} = W(\vx^{\le\alpha}, \vs^{\le\alpha})\in\cY^{\alpha n} $ in the prefix since he himself designed $ \vs^{\le\alpha}\in\cS^{\alpha n} $ and the channel law is given by a deterministic function $ W\colon\cX\times\cS\to\cY $. 
Given $ \vy^{\le\alpha} $, he can further compute the posterior distribution of $ (\bfm,\bfr)\in\cM\times\cR $ conditioned on $ \vy^{\le\alpha} $. 
In fact, he can compute the joint distribution $ P_{(\bfm,\bfr),\bfy^{\le\alpha}}\in\Delta((\cM\times\cR)\times\cY^{\alpha n}) $ as follows:
\begin{align}
P_{(\bfm,\bfr),\vbfy^{\le\alpha}}((m,r),\vy^{\le\alpha})
&= \frac{1}{MN}\sum_{(\vs(1),\cdots,\vs(\alpha n))\in\cS^{\alpha n}} \prod_{i = 1}^n 
\paren{V_{\bfs|\bfx,\bfu^{\le\alpha}}(\vs(i)|\vx_{(m,r)}(i),\vu^{\le\alpha}(i))W_{\bfy|\bfx,\bfs}(\vy(i)|\vx_{(m,r)}(i),\vs(i))} . \notag 
\end{align}
The conditional distribution $ P_{(\bfm,\bfr)|\bfy^{\le\alpha}} $ can therefore be computed from the above joint distribution. 
Define then $ P'_{(\bfm,\bfr)|\bfy^{\le\alpha} = \vy^{\le\alpha}} $ as the restriction of $ P_{(\bfm,\bfr)|\bfy^{\le\alpha} = \vy^{\le\alpha}} $ to the subset $ \cM'\times\cR' $:
\begin{align}
P'_{(\bfm,\bfr)|\bfy^{\le\alpha} = \vy^{\le\alpha}}( (m,r))
&= \frac{1}{Z(\vy^{\le\alpha})} P_{(\bfm,\bfr)|\bfy^{\le\alpha} = \vy^{\le\alpha}} ( (m,r)) \indicator{(m,r)\in\cM'\times\cR'} \label{eqn:posterior-distr-prime} 
\end{align}
and $ Z(\vy^{\le\alpha}) $ is a normalizing factor
\begin{align}
Z(\vy^{\le\alpha}) & \coloneqq \sum_{(m,r)\in\cM'\times\cR'} P_{(\bfm,\bfr)|\bfy^{\le\alpha} = \vy^{\le\alpha}} ( (m,r)) . \notag 
\end{align}
Note that $ P'_{(\bfm,\bfr)|\bfy^{\le\alpha} = \vy^{\le\alpha}} $ is in fact consistent with Bob's posterior distribution on $ (\bfm,\bfr) $ given his observation $ \vy^{\le\alpha} $ since we assume that Bob knows $ \cM'\times\cR' $.

\subsubsection{Symmetrization in the suffix}
\label{sec:attack-symm}
James samples $ (\bfm',\bfr')\in\cM\times\cR $ as follows
\begin{align}
(\bfm',\bfr')\sim P'_{(\bfm,\bfr)|\bfy^{\le\alpha} = \vy^{\le\alpha}} . \notag 
\end{align}
Using this spoofing message-key pair, James computes the encoding $ \enc(\bfm',\bfr') = \vx_{(\bfm',\bfr')}\in\cX^n $. 
Oftentimes we will write $ \vbfx' $ for simplicity.

Let $ \vv^{>\alpha}\in(\cU^{>\alpha})^{(1-\alpha)n} $ to be determined below. 
For each $ i = \alpha n+1,\alpha n+2,\cdots, n $, suppose $ (\vx(i), \vbfx'(i)) = (x,x') $ for some $ (x,x')\in\cX^2 $. 
Upon observing $ \vx(i) $, James dynamically assigns time-sharing symbols to each time index in a greedy manner. 
If it happens to be the case that $ \norminf{\type_{\vx^{>\alpha},\vbfx'^{>\alpha}} - \wh{P}_{\bfx,\bfx'}}\le\slk{CP} $, then after $ i=n $, each entry of the time-sharing sequence $ \vv^{>\alpha} $ will be assigned values from $ [K'] $ so that 
\begin{align}
\norminf{\type_{\vv^{>\alpha},\vx^{>\alpha},\vbfx'^{>\alpha}} - P_{\bfv} P_{\bfx|\bfv}^{\ot2}} &\le \slk{CP} . \label{eqn:close-to-cp} 
\end{align}
Then James samples $ \vbfs(j) $ as follows
\begin{align}
\vbfs(j) &\sim V_{\bfs|\bfx = \vx(j), \bfx' = \vbfx'(j),\bfu^{>\alpha} = u}. \notag 
\end{align}
If James is not able to assign values to $ \vv^{>\alpha} $ in a way that is consistent with $ P_{\bfv}P_{\bfx|\bfv}^{\ot2} $, then this indicates that $ \norminf{\type_{\vx^{>\alpha},\vbfx'^{>\alpha}} - \wh{P}_{\bfx,\bfx'}}>\slk{CP} $ and he declares an attack failure.

If at any time $ i<n $ in any stage of the attack, James runs out of his jamming budget, i.e., the jamming sequence $ (\vbfs(1),\cdots,\vbfs(i)) $ already violates the power constraint $ \stcstr $, then he declares an attack failure.

\subsection{Analysis of the babble-only attack}
The error analysis of the babble-only attack described in \Cref{sec:converse-babble-only} will be completely subsumed by the following analysis for the babble-and-push attack (described in \Cref{sec:converse-twostage-attack}) by setting $ \alpha = 1 $. 
We therefore omit the analysis of the former and proceed with that of the latter.

\subsection{Analysis of the babble-and-push attack}

Define
\begin{align}
\cE_0 &\coloneqq \curbrkt{(\bfm,\bfr)\in\cM'\times\cR'} . \label{eqn:def-e0} 
\end{align}
Since a message-key pair uniquely specifies a codeword, by \Cref{eqn:cc-reduction-guarantee}, we have the following lemma. 
\begin{lemma}
For $ \cE_0 $ defined in \Cref{eqn:def-e0}, it holds that 
$ \prob{\cE_0} \ge c $ where $ c>0 $ is given by \Cref{eqn:cc-reduction-guarantee}. 
\end{lemma}
In the rest of the analysis, we will condition on $ \cE_0 $ and suppress the notation for conditioning.

Conditioned on $ \cE_0 $, the joint distribution of the triple $ (\bfm,\vbfx,\vbfy^{\le\alpha}) $ is given by
\begin{align}
P_{\bfm,\vbfx,\vbfy^{\le\alpha}}(m,\vx,\vy^{\le\alpha}) 
&= \frac{1}{Z} \sum_{\vs^{\le\alpha}\in\cS^{n\alpha}} \frac{1}{|\cM|} P_{\vbfx|\bfm}(\vx|m) V_{\bfs|\bfx,\bfu^{\le\alpha}}(\vs(i)|\vx(i),\vu^{\le\alpha}) W_{\bfy|\bfx,\bfs}(\vy(i)|\vx(i),\vs^{\le\alpha}(i)) \indicator{\vx\in\cC'} \label{eqn:joint-distr-prefix} 
\end{align}
where 
\begin{align}
Z &\coloneqq \probover{(\bfm,\bfr)\sim\cM\times\cR}{\vx_{(\bfm,\bfr)}\in\cC'} = \prob{\cE_0} . \notag 
\end{align}
W.r.t.\ the above joint distribution, define 
\begin{align}
\cA_1\coloneqq\curbrkt{\vy^{\le\alpha}\in\cY^{n\alpha}:H(\bfm|\vbfy^{\le\alpha} = \vy^{\le\alpha})\ge n\delta/4},\quad 
\cE_1\coloneqq\curbrkt{\vbfy^{\le\alpha}\in\cA_1}.  \label{eqn:def-e1} 
\end{align}

We will exhibit a lower bound on the probability of $ \cE_1 $, that is, we will show that with at least a constant probability, the the transmitted message $\bfm$ has a nontrivial amount of residual entropy given the received vector $ \vbfy^{\le\alpha} $ in the first stage. 

\begin{lemma}
\label{lem:lb-e1}
For $ \cE_0, \cE_1 $ defined in \Cref{eqn:def-e0,eqn:def-e1}, it holds that 
\begin{align}
\prob{\cE_1 \condon \cE_0} &\ge \frac{\delta}{4} - \frac{1}{n} \log\frac{1}{c} , \notag 
\end{align}
where $ c>0 $ is given by \Cref{eqn:cc-reduction-guarantee}.
\end{lemma}

\begin{proof}
By the Data Processing Inequality (\Cref{lem:dpi}), 
\begin{align}
I(\bfm;\vbfy^{\le\alpha}) &\le I(\vbfx^{\le\alpha};\vbfy^{\le\alpha}) \notag \\
&\le \sum_{u\in\cU^{\le\alpha}} I(\vbfx^{(u)};\vbfy^{(u)}) \notag \\
&\le \sum_{u\in\cU^{\le\alpha}} n\eps I(\bfx_u^{\le\alpha};\bfy_u^{\le\alpha}) \label{eqn:p-xu-yu} \\
&= n\alpha \sum_{u\in\cU^{\le\alpha}}\frac{1}{\alpha K} I(\bfx_u^{\le\alpha};\bfy_u^{\le\alpha}) \notag \\
&= n\alpha I(\bfx^{\le\alpha};\bfy^{\le\alpha}|\bfu^{\le\alpha}). \label{eqn:x-y-given-u} 
\end{align}
All information measures involving $ \bfm,\vbfx,\vbfy^{\le\alpha} $ are computed according to \Cref{eqn:joint-distr-prefix}. 
In \Cref{eqn:p-xu-yu}, $ (\bfx_u^{\le\alpha},\bfy_u^{\le\alpha}) $ is distributed according to $ P_{\bfx_u^{\le\alpha},\bfy_u^{\le\alpha}}\in\Delta(\cX\times\cY) $ defined as
\begin{align}
P_{\bfx_u^{\le\alpha},\bfy_u^{\le\alpha}}(x,y)
&= \sum_{s\in\cS}P_{\bfx|\bfu = u}(x)V_{\bfs|\bfx,\bfu = u}(s|x)W_{\bfy|\bfx,\bfs}(y|x,s). \notag 
\end{align}
In \Cref{eqn:x-y-given-u}, the mutual information is evaluated according to $ P_{\bfu^{\le\alpha},\bfx^{\le\alpha},\bfy^{\le\alpha}}\in\Delta(\cU^{\le\alpha}\times\cX\times\cY) $ defined as
\begin{align}
P_{\bfu^{\le\alpha},\bfx^{\le\alpha},\bfy^{\le\alpha}}(u,x,y)
&= \sum_{s\in\cS} \frac{1}{\alpha K} P_{\bfx|\bfu}(x|u) V_{\bfs|\bfx,\bfu}(s|x,u) W_{\bfy|\bfx,\bfs}(y|x,s). \notag 
\end{align}
Therefore, 
\begin{align}
H(\bfm|\vbfy^{\le\alpha}) 
= H(\bfm) - I(\bfm;\vbfy^{\le\alpha}) 
\ge nR - \log\frac{1}{c} - n\alpha I(\bfx^{\le\alpha};\bfy^{\le\alpha}|\bfu^{\le\alpha}) 
\ge n\delta /2 - \log\frac{1}{c} . \notag 
\end{align}
The first inequality follows since conditioned on $ \cE_0 $, $ \bfm $ is uniformly distributed in $ \cM' $, and we use the bound $ |\cM'| \ge c2^{nR} $. 
To see the latter bound, one simply notes that $ |\cM'| $ is obviously at least the LHS of \Cref{eqn:cc-reduction-guarantee} since the inner summation (over $ \vx $) of the LHS is at most $1$. 
The second inequality above follows from the choice of the rate (see \Cref{eqn:converse-rate-assumption}) and the guarantee of the subcode (see \Cref{eqn:rate-guarantee}). 

We now argue that the event $ \cE_1 $ has a nontrivial probability to happen.
This follows from a Markov-type argument. 
\begin{align}
n\delta/2 - \log\frac{1}{c} &\le H(\bfm|\vbfy^{\le\alpha}) = \sum_{\vy^{\le\alpha}\in\cY^{n\alpha}} \prob{\vbfy^{\le\alpha} = \vy^{\le\alpha}}H(\bfm|\vbfy^{\le\alpha} = \vy^{\le\alpha}) \notag \\
&= \sum_{\vy^{\le\alpha}\in\cA_1} + \sum_{\vy^{\le\alpha}\not\in\cA_1} \prob{\vbfy^{\le\alpha} = \vy^{\le\alpha}}H(\bfm|\vbfy^{\le\alpha} = \vy^{\le\alpha}) \notag \\
&\le \prob{\cE_1\condon\cE_0}\cdot nR + (1 - \prob{\cE_1\condon\cE_0})\cdot{n\delta}/{4}
\le \prob{\cE_1\condon\cE_0}\cdot n + {n\delta}/{4} . \notag
\end{align} 
Therefore
$\prob{\cE_1\condon\cE_0} \ge \delta/4$. 
\end{proof}

Let $ (\bfm',\bfr') $ be the message-key pair that James sampled during the second stage of his attack. 
For James attack to be successful, it is required that the spoofing message $ \bfm' $ is different from the Alice's chosen message $\bfm$. 
This will indeed be the case with a constant probability. 

It was shown in \cite{dey-causal-it2013} that as long as the entropy of a source distribution is high, then i.i.d.\ samples from that distribution are unlikely to collide. 
\begin{lemma}[\cite{dey-causal-it2013}]\label{lem:fano-no-collision}
Let $\cV$ be a finite set and $ P\in\Delta(\cV) $. 
Let $ \bfv_1,\cdots,\bfv_k $ be i.i.d.\ samples from $ P $. Then
\begin{align}
\prob{\card{\curbrkt{\bfv_1,\cdots,\bfv_k}} = k} \ge& \paren{\frac{H(P) - 1 - \log k}{\log|\cV|}}^{k - 1}. \notag 
\end{align}
\end{lemma}

Let 
\begin{align}
\cE_2 &\coloneqq \curbrkt{\bfm\ne\bfm'} , \label{eqn:def-e2} \\
\cE_3 &\coloneqq\curbrkt{\norminf{\type_{\vx_{(\bfm,\bfr)}^{>\alpha}, \vx_{(\bfm',\bfr')}^{>\alpha}} - \wh{P}_{\bfx,\bfx'}}\le\slk{CP}}. \label{eqn:def-e3}  
\end{align}

Since conditioned on $ \vbfy^{\le\alpha} $, $ \bfm $ and $ \bfm' $ are i.i.d.\ according to $ P'_{(\bfm,\bfr)|\vbfy^{\le\alpha}} $ (cf.\ \Cref{eqn:posterior-distr-prime}), \Cref{lem:fano-no-collision} immediately yields a lower bound on $ \prob{\cE_2\condon\cE_0\cap\cE_1} $. 

\begin{lemma}
\label{lem:lb-e2}
For $ \cE_0, \cE_1,\cE_2 $ defined in \Cref{eqn:def-e0,eqn:def-e1,eqn:def-e2}, respectively, it holds that
\begin{align}
\prob{\cE_2\condon\cE_0\cap\cE_1} &\ge {\frac{n\delta/2 - 1 - \log 2}{\log M}} \ge {\frac{\delta}{2} - \frac{1}{n}} . \notag 
\end{align}
\end{lemma}

We then would like to lower bound $ \prob{\cE_2\cap\cE_3\condon\cE_0\cap\cE_1} $. 
To this end, we first recall a useful combinatorial theorem (which we call the \emph{Generalized Plotkin Bound}) proved in \cite{wbbj-2019-omni-avc}. 
It concerns a structural property of a generic set of sequences. 

\begin{theorem}[Generalized Plotkin bound, \cite{wbbj-2019-omni-avc}]
\label{thm:gen-plotkin}
Let $ \cX $ be a finite set and $ \slk{input},\slk{CP}>0 $ be small constants. 
There exists a large constant $ N_{\textnormal{Plokin}}\in\bZ_{\ge1} $ depends on $ \slk{input},\slk{CP},|\cX| $ but not on $ n $ such that
for any sufficiently large $ n\in\bZ_{\ge1} $ and any set $ \curbrkt{\vx_1,\cdots,\vx_M}\subset\cX^n $ of distinct vectors satisfying
\begin{enumerate}
	\item $ \norminf{\type_{\vx} - P_\bfx}\le\slk{input} $ for some $ P_\bfx\in\Delta(\cX) $;

	\item $ M\ge N_{\textnormal{Plotkin}} $, 
\end{enumerate}
there must exist $ 1\le i\ne j\le M $ and $ P_{\bfx,\bfx'}\in\cp(P_\bfx) $ such that 
\begin{align}
\norminf{\type_{\vx_i,\vx_j} - P_{\bfx,\bfx'}} &\le \slk{CP} . \notag 
\end{align}
\end{theorem}

We now use \Cref{thm:gen-plotkin} to prove a lower bound on $ \prob{\cE_2\cap\cE_3\condon\cE_1} $. 

\begin{lemma}
\label{lem:lb-e2-e3}
For $ \cE_0,\cE_1,\cE_2,\cE_3 $ defined in \Cref{eqn:def-e0,eqn:def-e1,eqn:def-e2,eqn:def-e3}, respectively, it holds that
\begin{align}
\prob{\cE_2\cap\cE_3\condon\cE_0\cap\cE_1} &\ge \frac{1}{N^2} \paren{\frac{\delta}{2} - \frac{1+\log N}{n}}^{N-1} , \notag 
\end{align}
where $ N = N_{\textnormal{Plotkin}} \in\bZ_{\ge1} $ is given by \Cref{thm:gen-plotkin}. 
In particular $ N $ depends only on $ \slk{input},\slk{CP},|\cX| $, but not on $n$. 
\end{lemma}

\begin{proof}
Let $ N\in\bZ_{\ge1} $ be a sufficiently large constant to be specified later and $ (\bfm_1,\bfr_1),\cdots,(\bfm_N,\bfr_N) $ be i.i.d.\ according to $ P'_{(\bfm,\bfr)|\bfy^{\le\alpha}=\vy^{\le\alpha}} $. 
Define
\begin{align}
\cE_2' &\coloneqq \curbrkt{\card{{\curbrkt{\bfm_1,\cdots,\bfm_N}}} = N} , \quad
\cE_{i,j} \coloneqq \curbrkt{\bfm_i\ne\bfm_j} \cap \curbrkt{\norminf{\type_{\vx_{(\bfm_i,\bfr_i)}^{>\alpha},\vx_{(\bfm_j,\bfr_j)}^{>\alpha}} - \wh{P}_{\bfx,\bfx'}}\le\slk{CP}} . \notag 
\end{align}
Now, for one thing, we have by the union bound
\begin{align}
\prob{\cE_2'\cap\bigcup_{1\le i,j\le N}\cE_{i,j}\condon\cE_0\cap\cE_1} &\le N^2 \prob{\cE_2\cap\cE_3\condon\cE_0\cap\cE_1} . \label{eqn:e23-1} 
\end{align}
For another thing, 
\begin{align}
\prob{\cE_2'\cap\bigcup_{1\le i,j\le N}\cE_{i,j}\condon\cE_0\cap\cE_1}
&= \prob{\cE_2'\condon\cE_0\cap\cE_1} \prob{\bigcup_{1\le i,j\le N}\cE_{i,j}\condon\cE_0\cap\cE_1\cap\cE_2'}
= \prob{\cE_2'\condon\cE_0\cap\cE_1}, \label{eqn:e23-2} 
\end{align}
where the last equality follows since by the Generalized Plotkin bound (\Cref{thm:gen-plotkin}), the event $ \bigcup_{i,j=1}^N\cE_{i,j} $ happens with probability $1$ given $ \cE_0\cap\cE_1\cap\cE_2' $. 
Specifically, once $ N $ is chosen to be larger than $ N_{\textnormal{Plotkin}} $ given by \Cref{thm:gen-plotkin}, $ \cE_1 $ implies the existence of a pair of (distinct) codewords whose joint type is approximately completely positive. 

By \Cref{lem:fano-no-collision}, 
\begin{align}
\prob{\cE_2'\condon\cE_0\cap\cE_1} &\ge \paren{\frac{n\delta/2 - 1 - \log N}{\log M}}^{N-1} . \label{eqn:e23-3} 
\end{align}
Therefore, \Cref{eqn:e23-1,eqn:e23-2,eqn:e23-3} imply the following lower bound
\begin{align}
\prob{\cE_2\cap\cE_3\condon\cE_0\cap\cE_1} &\ge \frac{1}{N^2} \paren{\frac{n\delta/2 - 1 - \log N}{\log M}}^{N-1}
\ge \frac{1}{N^2} \paren{\frac{\delta}{2} - \frac{1+\log N}{n}}^{N-1} , \notag 
\end{align}
as promised in \Cref{lem:lb-e2-e3}. 
\end{proof}

Let 
\begin{align}
\cE_4 &\coloneqq \curbrkt{ \type_{\vbfs} \in \stcstr } . \label{eqn:def-e4} 
\end{align}

\begin{lemma}[Concentration of jamming cost]
\label{lem:concentrate-jam-cost}
For $ \cE_0,\cE_1,\cE_2,\cE_3,\cE_4 $ defined in \Cref{eqn:def-e0,eqn:def-e1,eqn:def-e2,eqn:def-e3,eqn:def-e4}, respectively, it holds that
\begin{align}
\prob{\cE_4\condon\cE_0\cap\cE_1\cap\cE_2\cap\cE_3} &\ge 1-e^{-\Omega(\delta^2n)}. \notag 
\end{align}
\end{lemma}

\begin{proof}
Conditioned on $ \cE_1\cap\cE_2\cap\cE_3 $, during the symmetrization phase of the attack (\Cref{sec:attack-symm}), James is able to learn a time-sharing sequence $ \vv^{>\alpha}\in(\cU^{>\alpha})^{(1-\alpha)n} $ that is consistent with the completely positive structure of $ \wh{P}_{\bfx,\bfx'} $. 
That is, \Cref{eqn:close-to-cp} holds. 
This implies that the cost of the suffix jamming vector $ \vbfs^{>\alpha} $ is feasible with high probability. 
Indeed, by the Chernoff bound (\Cref{lem:chernoff}), 
\begin{align}
& \prob{\norminf{\type_{\vbfs^{>\alpha}|\vx_{(\bfm,\bfr)}^{>\alpha},\vx_{(\bfm',\bfr')}^{>\alpha},\vv^{>\alpha}} - V_{\bfs|\bfx,\bfx',\bfu^{>\alpha}}} > \frac{\delta}{8} \condon \cE_0\cap\cE_1\cap\cE_2\cap\cE_3} \notag \\
&= \prob{\exists (v,x,x')\in\cU^{>\alpha}\times\cX^2, \norminf{\type_{\vbfs^{>\alpha}|\vx_{(\bfm,\bfr)}^{>\alpha},\vx_{(\bfm',\bfr')}^{>\alpha},\vv^{>\alpha}}(\cdot|x,x',v) - V_{\bfs|\bfx ,\bfx' ,\bfu^{>\alpha} }(\cdot|x,x',v)} \ge \frac{\delta}{8} \condon\cE_0\cap\cE_1\cap\cE_2\cap\cE_3} \notag \\
&\le \sum_{(v,x,x')\in\cU^{>\alpha}\times\cX^2} \sum_{s\in\cS} \prob{\abs{ \type_{\vbfs^{>\alpha}|\vx_{(\bfm,\bfr)}^{>\alpha},\vx_{(\bfm',\bfr')}^{>\alpha},\vv^{>\alpha}}(s|x,x',v) - V_{\bfs |\bfx ,\bfx',\bfu^{>\alpha} }(s|x,x',v) } > \frac{\delta}{8}V_{\bfs |\bfx ,\bfx',\bfu^{>\alpha} }(s|x,x',v) \condon\cE_0\cap\cE_1\cap\cE_2\cap\cE_3} \notag \\
&\le \sum_{(v,x,x',s)\in\cU^{>\alpha}\times\cX^2\times\cS} 2\exp\paren{-\frac{(\delta/8)^2}{3}nP_\bfv(v)P_{\bfx|\bfv}(x|v)P_{\bfx|\bfv}(x'|v)V_{\bfs|\bfx,\bfx',\bfu^{>\alpha}}(s|x,x',v)} \notag \\ 
&= e^{-\Omega(\delta^2n)} . \label{eqn:jam-cost-suff} 
\end{align}
Similarly, the jamming cost in the prefix is also bounded with high probability
\begin{align}
\prob{ \norminf{ \type_{\vbfs^{\le\alpha}|\vx_{(\bfm,\bfr)}^{\le\alpha}, \vu^{\le\alpha}} - V_{\bfs|\bfx,\bfu^{\le\alpha}} } > \frac{\delta}{8} } 
\le \sum_{(u,x,s)\in\cU^{\le\alpha}\times\cX\times\cS} 2 \exp\paren{-\frac{(\delta/8)^2}{3}n\eps V_{\bfs|\bfx,\bfu^{\le\alpha}}(s|x,u)} 
= e^{-\Omega(\delta^2n)} . \label{eqn:jam-cost-pref} 
\end{align}
By the triangle inequality, \Cref{eqn:jam-cost-suff,eqn:jam-cost-pref} imply
\begin{align}
\norminf{ \type_{\vu^{\le\alpha},\vx_{(\bfm,\bfr)}^{\le\alpha}}(\cdot|u) \type_{\vbfs^{\le\alpha}|\vx_{(\bfm,\bfr)}^{\le\alpha}, \vu^{\le\alpha}} - P_{\bfu^{\le\alpha}} P_{\bfx|\bfu^{\le\alpha}} V_{\bfs|\bfx,\bfu^{\le\alpha}} }
&\le \slk{input} + \delta/8 + \slk{input}\delta/8 \le \slk{input} + \delta/4, \label{eqn:cost-joint-pref} \\
\norminf{\type_{\vv^{>\alpha}} \type_{\vx_{(\bfm,\bfr)}^{>\alpha}, \vx_{(\bfm',\bfr')}^{>\alpha}|\vv^{>\alpha}} \type_{\vbfs^{>\alpha}|\vx_{(\bfm,\bfr)}^{>\alpha},\vx_{(\bfm,\bfr')}^{>\alpha},\vv^{>\alpha}} - P_{\bfv} P_{\bfx|\bfv}^{\ot2} V_{\bfs|\bfx,\bfx',\bfu^{>\alpha}}} 
&\le \slk{CP} + \delta/8 + \slk{CP}\delta/8 \le \slk{CP} + \delta/4, \label{eqn:cost-joint-suff}
\end{align}
with high probability. 
Note that
\begin{align}
\type_{\vbfs} &= \alpha \sqrbrkt{\type_{\vu^{\le\alpha}} \type_{\vx_{(\bfm,\bfr)}^{\le\alpha}|\vu^{\le\alpha}} \type_{\vbfs^{\le\alpha}|\vx_{(\bfm,\bfr)}^{\le\alpha}, \vu^{\le\alpha}} }_{\bfs^{\le\alpha}} 
+ (1-\alpha) \sqrbrkt{ \type_{\vv^{>\alpha}} \type_{\vx_{(\bfm,\bfr)}^{>\alpha}, \vx_{(\bfm',\bfr')}^{>\alpha}|\vv^{>\alpha}} \type_{\vbfs^{>\alpha}|\vx_{(\bfm,\bfr)}^{>\alpha},\vx_{(\bfm,\bfr')}^{>\alpha},\vv^{>\alpha}} }_{\bfs^{>\alpha}} . \notag 
\end{align}
Define
\begin{align}
P_{\bfs} &\coloneqq \alpha \sqrbrkt{P_{\bfu^{\le\alpha}} P_{\bfx|\bfu^{\le\alpha}} V_{\bfs|\bfx,\bfu^{\le\alpha}} }_{\bfs^{\le\alpha}} + (1-\alpha) \sqrbrkt{P_{\bfv} P_{\bfx|\bfv}^{\ot2} V_{\bfs|\bfx,\bfx',\bfu^{>\alpha}}}_{\bfs^{>\alpha}} . \notag 
\end{align}
\Cref{eqn:cost-joint-pref,eqn:cost-joint-suff} further imply
\begin{align}
\norminf{ \type_{\vbfs} - P_{\bfs} } &\le |\cU^{\le\alpha}||\cX||\cS|(\slk{input}+\delta/4) + |\cU^{>\alpha}||\cX|^2|\cS|(\slk{CP} + \delta/4). \label{eqn:bound-cost} 
\end{align}
Since $ P_{\bfs}\in\interior_{\slk{state}}(\stcstr) $ and $ \slk{state} $ was set to be larger than the RHS of \Cref{eqn:bound-cost}, we have $ \type_{\vbfs}\in\stcstr $ with probability $ 1-e^{-\Omega(\delta^2n)} $ by the triangle inequality. 
\end{proof}

\begin{lemma}[Lower bound on error probability]
\label{lem:lb-pe}
For any code with a (potentially stochastic) encoder-decoder pair $ (\enc,\dec) $ and the babble-and-push attack $\jam$ described in \Cref{sec:converse-twostage-attack}, the average error probability defined in \Cref{eqn:def-pe-avg} is lower bounded as follows:
\begin{align}
P_{\e,\avg}(\enc,\dec,\jam)
&\ge \frac{c}{2}\cdot\paren{\frac{\delta}{4} - \frac{1}{n} \log\frac{1}{c}}\cdot \frac{1}{N^2}\paren{\frac{\delta}{2} - \frac{1+\log N}{n}}^{N-1}\cdot \paren{1-e^{-\Omega(\delta^2n)}} , \notag 
\end{align}
where $ c>0 $ is given by \Cref{eqn:cc-reduction-guarantee} and 
$ N = N_{\textnormal{Plotkin}} \in\bZ_{\ge1} $ is given by \Cref{thm:gen-plotkin}, both independent of $n$. 
\end{lemma}

\begin{proof}
Since conditioned on $ \vbfy^{\le\alpha} $, $ (\bfm,\bfr') $ and $ (\bfm',\bfr') $ are i.i.d.\ (more specifically, independent and both following the posterior distribution $ P'_{(\bfm,\bfr)|\vbfy^{\le\alpha} = \vy^{\le\alpha}} $), we have
\begin{align}
P_{\vbfy^{\le\alpha}, (\bfm,\bfr),(\bfm',\bfr')}(\vy^{\le\alpha},(m,r),(m',r')) 
=& P_{\vbfy^{\le\alpha}}(\vy^{\le\alpha}) P_{(\bfm,\bfr), (\bfm',\bfr') | \vbfy^{\le\alpha}}((m,r),(m',r')|\vy^{\le\alpha}) \notag \\
=& P_{\vbfy^{\le\alpha}}(\vy^{\le\alpha}) P'_{(\bfm,\bfr)  | \vbfy^{\le\alpha}}((m,r)|\vy^{\le\alpha}) P'_{(\bfm,\bfr)  | \vbfy^{\le\alpha}}((m',r')|\vy^{\le\alpha}) \notag \\
=& P_{\vbfy^{\le\alpha}}(\vy^{\le\alpha}) P_{ (\bfm',\bfr'),(\bfm,\bfr) | \vbfy^{\le\alpha} }((m,r),(m',r')|\vy^{\le\alpha}) \notag \\
=& P_{\vbfy^{\le\alpha}, (\bfm',\bfr'), (\bfm,\bfr)}(\vy^{\le\alpha},(m,r),(m',r')). \label{eqn:first-exchangeable}
\end{align}

We also claim that conditioned on $ \cE_0\cap\cE_1\cap\cE_2\cap\cE_3 $, 
\begin{align}
P_{\vbfy^{>\alpha}|\vbfy^{\le\alpha},(\bfm,\bfr),(\bfm',\bfr')} = P_{\vbfy^{>\alpha}|\vbfy^{\le\alpha},(\bfm',\bfr'),(\bfm,\bfr)}. 
\label{eqn:second-exchangeable}
\end{align}
To see this, by Bayes' theorem, 
\begin{align}
P_{\vbfy^{>\alpha}|\vbfy^{\le\alpha},(\bfm,\bfr),(\bfm',\bfr')} 
&= P_{\vbfy^{>\alpha},(\bfm,\bfr),(\bfm',\bfr')|\vbfy^{\le\alpha}}/P_{(\bfm,\bfr),(\bfm',\bfr')|\vbfy^{\le\alpha}}. \label{eqn:second-bayes} 
\end{align}
The distribution in the denominator, as we just argued (\Cref{eqn:first-exchangeable}), is conditionally exchangeable. 
For the numerator, we have
\begin{align}
& P_{\vbfy^{>\alpha},(\bfm,\bfr),(\bfm',\bfr')|\vbfy^{\le\alpha}}(\vy^{>\alpha},(m,r),(m',r')|\vy^{\le\alpha}) \notag \\
&= P'_{(\bfm,\bfr)|\vbfy}((m,r)|\vy^{\le\alpha})P'_{(\bfm,\bfr)|\vbfy}((m',r')|\vy^{\le\alpha}) \notag \\
&\qquad \sum_{\vs\in\cS^n} \prod_{i=\alpha n+1}^n \paren{ V_{\bfs|\bfx,\bfx',\bfu^{>\alpha}}(\vs(i)|\vx_{(m,r)}(i),\vx_{(m',r')}(i),\vv^{>\alpha}(i)) W_{\bfy|\bfx,\bfs}(\vy(i)|\vx_{(m,r)}(i),\vs(i)) } \notag \\
&= P'_{(\bfm,\bfr)|\vbfy}((m,r)|\vy^{\le\alpha})P'_{(\bfm,\bfr)|\vbfy}((m',r')|\vy^{\le\alpha}) \notag \\
&\qquad \prod_{i=\alpha n+1}^n \sum_{\vs(i)\in\cS} \paren{ V_{\bfs|\bfx,\bfx',\bfu^{>\alpha}}(\vs(i)|\vx_{(m,r)}(i),\vx_{(m',r')}(i),\vv^{>\alpha}(i)) W_{\bfy|\bfx,\bfs}(\vy(i)|\vx_{(m,r)}(i),\vs(i)) } \notag \\
&= P'_{(\bfm,\bfr)|\vbfy}((m',r')|\vy^{\le\alpha})P'_{(\bfm,\bfr)|\vbfy}((m,r)|\vy^{\le\alpha}) \notag \\
&\qquad \prod_{i=\alpha n+1}^n \sum_{\vs(i)\in\cS} \paren{ V_{\bfs|\bfx,\bfx',\bfu^{>\alpha}}(\vs(i)|\vx_{(m',r')}(i),\vx_{(m,r)}(i),\vv^{>\alpha}(i)) W_{\bfy|\bfx,\bfs}(\vy(i)|\vx_{(m',r')}(i),\vs(i)) } \label{eqn:use-symm} \\
&= P_{\vbfy^{>\alpha},(\bfm,\bfr),(\bfm',\bfr')|\vbfy^{\le\alpha}}(\vy^{>\alpha},(m',r'),(m,r)|\vy^{\le\alpha}). \label{eqn:second-exchangeable-pf}
\end{align}
where \Cref{eqn:use-symm} follows since $ V_{\bfs|\bfx,\bfx',\bfu^{>\alpha}}\in\cV $ (cf.\ \Cref{def:symm-dist}) and therefore satisfies the symmetrization equation. 
Combining \Cref{eqn:first-exchangeable,eqn:second-bayes,eqn:second-exchangeable-pf} proves \Cref{eqn:second-exchangeable}. 

Now, consider the joint distribution of $ (\vbfy,(\bfm,\bfr),(\bfm',\bfr')) $ given $ \cE_0\cap\cE_1\cap\cE_2\cap\cE_3 $. 
\begin{align}
P_{ \vbfy,(\bfm,\bfr),(\bfm',\bfr') } &= P_{\vbfy^{\le\alpha}, \vbfy^{>\alpha}, (\bfm,\bfr), (\bfm',\bfr)} \notag \\
&= P_{\vbfy^{\le\alpha}, (\bfm,\bfr), (\bfm',\bfr')} P_{\vbfy^{>\alpha} | \vbfy^{\le\alpha}, (\bfm,\bfr), (\bfm',\bfr')} \notag \\
&= P_{\vbfy^{\le\alpha}, (\bfm',\bfr'), (\bfm,\bfr)} P_{\vbfy^{>\alpha} | \vbfy^{\le\alpha}, (\bfm',\bfr'), (\bfm,\bfr)} \notag \\
&= P_{\vbfy, (\bfm',\bfr'), (\bfm,\bfr)}. \notag 
\end{align}
Define $ Q_{\vbfy,(\bfm,\bfr),(\bfm',\bfr')} $ as a truncated version of $ P_{\vbfy,(\bfm,\bfr),(\bfm',\bfr')} $ in which $ \vy $ results only from feasible $\vs$. 
Specifically,
\begin{align}
& Q_{\vbfy,(\bfm,\bfr),(\bfm',\bfr')}(\vy,(m,r),(m',r')) \notag \\
&\coloneqq \frac{1}{|\cM||\cR|} \sum_{\vs^{\le\alpha}\in\cS^{\alpha n}} \prod_{i = 1}^{\alpha n} \sqrbrkt{V_{\bfs|\bfx,\bfu^{\le\alpha}}(\vs^{\le\alpha}(i)|\vx_{(m,r)}(i),\vu^{\le\alpha}(i)) W_{\bfy|\bfx,\bfs}(\vy^{\le\alpha}(i)|\vx_{(m,r)}(i),\vs^{\le\alpha}(i))}
P_{(\bfm,\bfr)|\vbfy^{\le\alpha}}( (m',r')|\vy^{\le\alpha}) \notag \\
& \qquad \sum_{\vs^{>\alpha}\in\cS^{(1-\alpha)n}} \prod_{i = \alpha n+1}^n \sqrbrkt{V_{\bfs|\bfx,\bfx',\bfu^{>\alpha}}(\vs^{>\alpha}(i)|\vx_{(m,r)}(i),\vx_{(m',r')}(i),\vv^{>\alpha}(i)) W_{\bfy|\bfx,\bfs}(\vy^{>\alpha}(i)|\vx_{(m,r)}(i),\vs^{>\alpha}(i))} \one_{\cE_4} . \label{eqn:def-q}
\end{align}
We also define $ \ol{Q}_{\vbfy,(\bfm,\bfr),(\bfm',\bfr')} $ as the same expression as \Cref{eqn:def-q} with $ \one_{\cE_4} $ replaced with $ \one_{\cE_4^c} $. 
Note that 
\begin{align}
{Q}_{\vbfy,(\bfm,\bfr),(\bfm',\bfr')} &= P_{\vbfy,(\bfm,\bfr),(\bfm',\bfr')} - \ol{Q}_{\vbfy,(\bfm,\bfr),(\bfm',\bfr')} . \notag 
\end{align}

Suppose $ \dec\colon\cY^n\to[M] $ is the (potentially stochastic) decoder associated with $\cC$. 
Let $ P_{\wh{\bfm}|\vbfy} $ denote the conditional distribution induced by $ \dec $. 
Then 
\begin{align}
&P_{\e,\avg}(\enc,\dec,\jam) 
= \sum_{(m,r),(m',r'),\vy} Q_{\vbfy,(\bfm,\bfr),(\bfm,\bfr')}(\vy,(m,r),(m',r')) \sum_{\wh{m}\ne m} P_{\wh{\bfm}|\vbfy}(\wh{m}|\vy) \notag \\
&\ge \sum_{(m,r),(m',r'),\vy} \paren{P_{\vbfy,(\bfm,\bfr),(\bfm,\bfr')}(\vy,(m,r),(m',r')) - \ol{Q}_{\vbfy,(\bfm,\bfr),(\bfm,\bfr')}(\vy,(m,r),(m',r'))} \one_{\cE_0\cap\cE_1\cap\cE_2\cap\cE_3} %\notag \\
\sum_{\wh{m}\ne m} P_{\wh{\bfm}|\vbfy}(\wh{m}|\vy) \notag \\
&= \sum_{(m,r),(m',r'),\vy} P_{\vbfy,(\bfm,\bfr),(\bfm,\bfr')}(\vy,(m,r),(m',r'))  \one_{\cE_0\cap\cE_1\cap\cE_2\cap\cE_3} \sum_{\wh{m}\ne m} P_{\wh{\bfm}|\vbfy}(\wh{m}|\vy) \notag \\
&\qquad - \sum_{(m,r),(m',r'),\vy} \ol{Q}_{\vbfy,(\bfm,\bfr),(\bfm,\bfr')}(\vy,(m,r),(m',r')) \one_{\cE_0\cap\cE_1\cap\cE_2\cap\cE_3} \sum_{\wh{m}\ne m} P_{\wh{\bfm}|\vbfy}(\wh{m}|\vy) . \notag 
\end{align}
The second term can be upper bounded as follows:
\begin{align}
& \sum_{(m,r),(m',r'),\vy} \ol{Q}_{\vbfy,(\bfm,\bfr),(\bfm,\bfr')}(\vy,(m,r),(m',r')) \one_{\cE_0\cap\cE_1\cap\cE_2\cap\cE_3} \sum_{\wh{m}\ne m} P_{\wh{\bfm}|\vbfy}(\wh{m}|\vy) \notag \\
&\le \sum_{(m,r),(m',r'),\vy} \ol{Q}_{\vbfy,(\bfm,\bfr),(\bfm,\bfr')}(\vy,(m,r),(m',r')) \one_{\cE_0\cap\cE_1\cap\cE_2\cap\cE_3} \notag \\
&= \prob{\cE_0\cap\cE_1\cap\cE_2\cap\cE_3\cap\cE_4^c} \notag \\
&= \prob{\cE_0\cap\cE_1\cap\cE_2\cap\cE_3} \prob{\cE_4^c\condon\cE_0\cap\cE_1\cap\cE_2\cap\cE_3} . \notag %\\
% &\le \prob{\cE_0\cap\cE_1\cap\cE_2\cap\cE_3} e^{-\Omega(\delta^2n)} . \notag 
\end{align}
For the first term, we have
\begin{align}
& \sum_{(m,r),(m',r'),\vy} P_{\vbfy,(\bfm,\bfr),(\bfm,\bfr')}(\vy,(m,r),(m',r'))  \one_{\cE_0\cap\cE_1\cap\cE_2\cap\cE_3} \sum_{\wh{m}\ne m} P_{\wh{\bfm}|\vbfy}(\wh{m}|\vy) \notag \\
&= \sum_{(m',r'),(m,r),\vy} P_{\vbfy,(\bfm,\bfr),(\bfm,\bfr')}(\vy,(m',r'),(m,r)) \one_{\cE_0\cap\cE_1\cap\cE_2\cap\cE_3} \sum_{\wh{m}\ne m'} P_{\wh{\bfm}|\vbfy}(\wh{m}|\vy) \notag \\
&= \sum_{(m,r),(m',r'),\vy} P_{\vbfy,(\bfm,\bfr),(\bfm,\bfr')}(\vy,(m,r),(m',r')) \one_{\cE_0\cap\cE_1\cap\cE_2\cap\cE_3} \sum_{\wh{m}\ne m'} P_{\wh{\bfm}|\vbfy}(\wh{m}|\vy). \notag
\end{align}
Therefore,
\begin{align}
& 2 \sum_{(m,r),(m',r'),\vy} P_{\vbfy,(\bfm,\bfr),(\bfm,\bfr')}(\vy,(m,r),(m',r'))  \one_{\cE_0\cap\cE_1\cap\cE_2\cap\cE_3} \sum_{\wh{m}\ne m} P_{\wh{\bfm}|\vbfy}(\wh{m}|\vy) \notag \\ 
&\ge \sum_{(m,r),(m',r'),\vy} P_{\vbfy,(\bfm,\bfr),(\bfm,\bfr')}(\vy,(m,r),(m',r')) \one_{\cE_0\cap\cE_1\cap\cE_2\cap\cE_3} \paren{ \sum_{\wh{m}\ne m} P_{\wh{\bfm}|\vbfy}(\wh{m}|\vy) + \sum_{\wh{m}\ne m'} P_{\wh{\bfm}|\vbfy}(\wh{m}|\vy) } . \label{eqn:pe-symm} 
\end{align}
Note that
\begin{align}
\sum_{\wh{m}\ne m} P_{\wh{\bfm}|\vbfy}(\wh{m}|\vy) + \sum_{\wh{m}\ne m'} P_{\wh{\bfm}|\vbfy}(\wh{m}|\vy) 
&= \sum_{\wh m} P_{\wh{\bfm}|\vbfy}(\wh{m}|\vy) \paren{ \indicator{\wh m\ne m} + \indicator{\wh m\ne m'} } \notag \\
&\ge \min_{\wh m}\curbrkt{\indicator{\wh m\ne m} + \indicator{\wh m\ne m'}} \notag \\
&= \indicator{m \ne m'}. \label{eqn:indicator-symm} 
\end{align}
Substituting \Cref{eqn:indicator-symm} into \Cref{eqn:pe-symm}, we get
\begin{align}
& 2 \sum_{(m,r),(m',r'),\vy} P_{\vbfy,(\bfm,\bfr),(\bfm,\bfr')}(\vy,(m,r),(m',r'))  \one_{\cE_0\cap\cE_1\cap\cE_2\cap\cE_3} \sum_{\wh{m}\ne m} P_{\wh{\bfm}|\vbfy}(\wh{m}|\vy) \notag \\ 
&\ge \sum_{(m,r),(m',r'),\vy} P_{\vbfy,(\bfm,\bfr),(\bfm,\bfr')}(\vy,(m,r),(m',r')) \one_{\cE_0\cap\cE_1\cap\cE_2\cap\cE_3} \indicator{m\ne m'} \notag \\
&= \sum_{(m,r),(m',r'),\vy} P_{\vbfy,(\bfm,\bfr),(\bfm,\bfr')}(\vy,(m,r),(m',r')) \one_{\cE_0\cap\cE_1\cap\cE_2\cap\cE_3} \label{eqn:remove-indicator} \\
&= \prob{\cE_0\cap\cE_1\cap\cE_2\cap\cE_3} . \notag 
\end{align}
In \Cref{eqn:remove-indicator}, we remove $ \indicator{m\ne m'} $ since the event coincides with $ \cE_2 $. 
Finally, the overall error probability is at least 
\begin{align}
P_{\e,\avg}(\enc,\dec,\jam) 
&\ge \frac{1}{2} \prob{\cE_0\cap\cE_1\cap\cE_2\cap\cE_3} - \prob{\cE_0\cap\cE_1\cap\cE_2\cap\cE_3} \prob{\cE_4^c\condon\cE_0\cap\cE_1\cap\cE_2\cap\cE_3} \notag \\
&\ge \frac{1}{2} \prob{\cE_0\cap\cE_1\cap\cE_2\cap\cE_3} \paren{1 - 2\cdot e^{-\Omega(\delta^2n)}} \notag \\
&= \frac{1}{2} \prob{\cE_0} \prob{\cE_1\condon\cE_0} \prob{\cE_2\cap\cE_3\condon\cE_0\cap\cE_1} \paren{1 - e^{-\Omega(\delta^2n)}} \notag \\
&\ge \frac{c}{2}\cdot \frac{\delta}{4}\cdot \frac{1}{N^2}\paren{\frac{\delta}{2} - \frac{1+\log N}{n}}^{N-1}\cdot \paren{1-e^{-\Omega(\delta^2n)}} . \notag 
\end{align}
This finishes the proof of \Cref{lem:lb-pe}. 
\end{proof}

% \section{Evaluation of the capacity expression for examples}
% \label{sec:examples}

% \subsection{Bit-flip channels}
% \label{sec:example-bitflip}

% \subsection{Erasure channels}
% \label{sec:example-erasure}

\appendices

% !TEX root = causal_general.tex

\section{Basic tools}

% \begin{definition}[Net]
% \label{def:net}
% Given a matrix space $ (\cX,\dist) $ and a constant $ \eta>0 $, an \emph{$ \eta $-net} $\cN$ of $\cX$ w.r.t.\ the metric $ \dist $ is a subset $ \cN\subset\cX $ satisfying: for every $ x\in\cX $, there is $ x'\in\cN $ such that $ \dist(x,x')\le\eta $. 
% \end{definition}

% \begin{lemma}[Size of a net]
% \label{lem:size-net}
% Let $ \cX $ be a finite set. 
% For any constant $ \eta>0 $, there is an $ \eta $-net of $ (\Delta(\cX),\norminf{\cdot}) $ of size at most $ \ceil{|\cX|/(2\eta)}^{|\cX|}\le\paren{|\cX|/(2\eta)+1}^{|\cX|} $. 
% \end{lemma}

We recall some useful facts.

\begin{lemma}[Markov]
\label{lem:markov}
If $X$ is a nonnegative random variable with mean $\expt{X}$ then for any $a > 0$ we have $\prob{X \ge a} \le {\expt{X}}/{a}$.
\end{lemma}

\begin{lemma}[Data processing]
\label{lem:dpi}
If $ X\to Y\to Z $ is a Markov chain, then $ I(X; Z)\le I( X; Y) $. 
\end{lemma}

\begin{lemma}[Chernoff]
\label{lem:chernoff}
Let $ X_1,\cdots,X_n $ be independent $ \{0,1\} $-valued random variables and $ S\coloneqq\sum_{i=1}^nX_i $. 
Then for any $ \delta\in[0,1] $, 
\begin{align}
\prob{\abs{S - \expt{S}} \ge \delta \expt{S}} &\le 2\exp\paren{-\frac{\delta^2}{3}\expt{S}} . \notag 
\end{align}
\end{lemma}

\bibliographystyle{alpha}
\bibliography{IEEEabrv,ref}

\end{document}